\documentclass[12pt]{article}
\usepackage{amsmath,amssymb,amsfonts,amsthm}

\usepackage{color}
\usepackage{multirow}
\usepackage{lscape}

\topmargin - 1.3 cm \hoffset -1 cm

\newcounter{item}[section]
\newcounter{kirshr}
\newcounter{kirsha}
\newcounter{kirshb}
\newenvironment{mysect}[1]{\vskip8pt\par\noindent\setcounter{item}{1}

\setcounter{equation}{0}{\large\bf\arabic{section}.  #1 }\vskip8pt\nopagebreak\par\nopagebreak }
{\stepcounter{section}\upshape\par}

\newtheorem{theorem}{Theorem}[section]

\newtheorem{lemma}[theorem]{Lemma}
\newtheorem{corollary}[theorem]{Corollary}
\newtheorem{proposition}[theorem]{Proposition}
\newtheorem{remark}[theorem]{Remark}

\newtheorem{definition}[theorem]{Definition}
\newcommand\overcirc[1]{\raisebox{10pt}{\tiny$\circ$}{\kern-6.5pt}\mbox{$#1$}}
\newcommand\undersym[2]{\raisebox{-6pt}{\tiny$#2$}{\kern-5pt}\mbox{$#1$}}

\begin{document}

\title{\bf Teleparallel Lgrange Geometry and a Unified Field Theory\footnote{arXiv:0905.0209
}}

\author{M. I. Wanas$^\dagger$, Nabil L. Youssef$^{\,\ddagger}$ and  A. M.  Sid-Ahmed$^{\natural}$\footnote{The authors are members of the Egyptian Relativity Group (ERG): www.erg.net.eg}}
\date{}
\maketitle

\vspace{-1.1cm}
\begin{center}
{$\dagger$ Department of Astronomy, Faculty of Science, Cairo
University\\ CTP of the British University in Egypt (BUE)}
\end{center}
\vspace{-0.9cm}
\begin{center}
 wanas@frcu.eun.egt
 \end{center}
\vspace{-0.6cm}
\begin{center}
{$\ddagger$ Department of Mathematics, Faculty of Science, Cairo
University}
\end{center}
\vspace{-0.9cm}
\begin{center}
{nyoussef@frcu.eun.eg, \,nlyoussef2003@yahoo.fr}
\end{center}
\vspace{-0.6cm}
\begin{center}
{$\natural$ Department of Mathematics, Faculty of Science, Cairo
University}\end{center}
\vspace{-0.9cm}
\begin{center}
{amrs@mailer.eun.eg, \,amrsidahmed@gmail.com}
\end{center}

\maketitle \vspace{-1cm}

\vspace{1.5cm} \maketitle
\smallskip

\noindent{\bf Abstract.} In this paper, we construct a field theory
unifying gravity and electromagnetism in the context of Extended
Absolute Parallelism (EAP-) geometry. This geometry combines, within
its structure, the geometric richness of the tangent bundle and the
mathematical simplicity of Absolute Parallelism (AP-) geometry. The
constructed field theory is a generalization of the Generalized
Field Theory (GFT) formulated by Mikhail and Wanas. The theory
obtained is purely geometric. The horizontal (resp. vertical) field
equations are derived by applying the Euler-Lagrange equations to an
appropriate horizontal (resp. vertical) scalar Lagrangian. The
symmetric part of the resulting horizontal (resp. vertical) field
equations gives rise to a generalized form of Einstein's field
equations in which the horizontal (resp. vertical) energy-momentum
tensor is purely geometric. The skew-symmetric part of the resulting
horizontal (resp. vertical) field equations gives rise to a
generalized form of Maxwell equations in which the electromagnetic
field is purely geometric. Some interesting special cases, which
reveal the role of the nonlinear connection in the obtained field
equations, are examined. Finally, the condition under which our
constructed field equations reduce to the GFT is explicitly
established.

\bigskip


\bigskip




\newpage

\begin{mysect}{Motivation and Introduction}

Although the general theory of relativity, constructed in a 4-dimensional Riemannian space, is the best known theory for studying gravitational
interactions, so far, it suffers from some problems. Examples of these problems are: the horizon problem, the initial singularity, the flatness of the
rotation curve of spiral galaxies \cite{gal}, the Pioneer 10, 11-anamoly \cite{MNJ} and the interpretation of supernovae type-Ia observation \cite{FTD}.
Some of these problems
are old, while others have been discovered in the last ten years or so. In the context of orthodox general relativity theory, there are no satisfactory
solutions for such problems.

\bigskip

The above mentioned problems may be due to a missing interaction that has no representative in Riemannian geometry. This may imply that
Riemannian geometry is inadequate for studying such problems, since it is limited to the case of a symmetric linear connection and a symmetric metric tensor. However,
some authors have suggested modifications of general relativity, retaining Riemannian geometry, by either:

\begin{description}
\item [(a)] Increasing the order of the Lagrangian used to construct the field equations in general relativity \cite{f(R)}.

\item [(b)] Using non-conventional equations of state \cite{DAJ}.

\item [(c)] Adding a term to the theory (i.e the cosmological term) preserving conservation \cite{zoo}.

\item [(d)] Adding a term to the theory violating conservation \cite{zoz}.

\item [(e)] Increasing the dimension of Riemannian space used (Kaluza-Klien-type theories) \hspace{- 0.1 cm}\cite{o}.

\end{description}

Other authors prefer to use more general geometric structures, other than Riemannian geometry, e.g. Riemann-Cartan geometry \cite{AE},
Absolute Parallelism (AP-) geometry (cf. \cite{FI}, \cite{b}, \cite{AMR}), Finsler geometry and its generalizations (cf. \cite{Gh}, \cite{MM}) and
generalized AP-geometry \cite{WNA}. The use of these structures has the advantage of probing the role of geometric entities other
than the curvature, e.g. non-symmetric linear connection and its torsion, in physical applications. This may illuminate the
role of such entities in physical phenomena which
have no satisfactory interpretation in orthodox general relativity.

\bigskip

As an example, some of the geometric advantages and physical achievements of the conventional AP-geometry are the following:
\begin{description}

\item [(a)] AP-geometry admits at least four built-in (natural)
affine connections, two of
which are non-symmetric and three of which have non-vanishing curvature. AP-geometry also admits tensors of third order, a number of second order skew and
symmetric tensors and a {\bf non-vanishing torsion} \cite{FI}.

\item [(b)] Electromagnetism can be successfully represented together with gravity \cite{aaa}.

\item [(c)] In four dimensions, the tetrad vector field defining the
geometric structure of AP-space is
used as fundamental variables in an attempt to
quantize gravity \cite{vv}.

\item [(d)] AP-geometry gives rise to a new interaction between the torsion of the background geometry and the spin of the moving particles \cite{msp},
which has been confirmed experimentally \cite{q}.

\item [(e)] It has been shown that there is a built-in quantum properties in any geometric structure with non-vanishing torsion \cite{magd}.

\end{description}

We consider the Generalized Field Theory (GFT) \cite{aaa} as a good example of using geometries with simultaneously non-vanishing curvature and torsion for the
purpose of unifying fundamental interactions. Applications of this theory show to what extent it is successful in unifying gravity with
electromagnetism (cf. \cite{aax}, \cite{MIW}, \cite{W}).
\bigskip

Recently, many attempts have been made to formulate Einstein's field equations and Maxwell's equations in the context of
Finsler geometry and generalized Lagrange geometry (\cite{KM}, \cite{Miron}, \cite{C}, \cite{M}). For example, Almost Finslerian Lagrange
(AFL-) spaces \cite{V} were used as a model for a theory of electromagnetism
in this wider framework \cite{xxx}, \cite{abc}.
Moreover, R. Miron developed a Lagrangian theory
of relativity, which is a more generalized version of the Finslerian theory of relativity. In his theory, four kinds of new Einstein-like equations
and two kinds of conservation-like laws are constructed \cite{V}.

\bigskip

The use of Riemannian geometry in applications explores the role
played by its symmetric linear connection (and its consequences) in
physical interactions. Similarly, the use of more general geometric
structures, with non-vanishing torsion, explores the role played by
the non-symmetric linear connection (and its consequences) in
physical applications.

\bigskip

One of the aims of the present work is to explore the role of the nonlinear connection in physical phenomena,
if any. As a first step to achieve our goal is to construct a field theory in spaces equipped with a nonlinear connection. For this to be done,
two of the authors of this paper have constructed a geometric structure called Extended Absolute Parallelism \linebreak (EAP-) geometry
\cite{EAP}. EAP-geometry combines, within its structure, the geometric richness of the tangent bundle \cite{NLY} and the mathematical simplicity of AP-geometry.
Consequently, it may have a potentially wider geometric and physical scope than AP-geometry.

\bigskip

In the present work, we are going to construct a version of the GFT within the context of EAP-geometry, performing a suitable generalization of the
scheme followed in the construction of the field equations of the GFT. Our theory is formulated on the tangent bundle $TM$ of
$M$, on the basis of Miron's approach to the geometry of tangent bundles. However,
the method of construction of the theory and its content are substantially different from those of Miron's.

\bigskip

The paper is organized in the following manner.
In section 1, we focus our attention on the fundamental concepts that will be needed in the sequel. We then discuss the properties of a unique metric $d$-connection on $TM$ which we refer to as the
natural metric $d$-connection, relative to which the generalized field equations are to be obtained.
In section 2, a survey of the basic definitions and concepts of EAP-geometry is given, followed, in section 3,
by some relations needed for the construction of the field equations. In section 4, the field equations are constructed.
We derive the horizontal generalized field equations
by applying a modified version of the Euler-Lagrange equations to a suitable (horizontal) scalar Lagrangian ${\cal H}$. The vertical generalized
field equations {\it in the general case}
are also found, again by applying the Euler-Lagrange equations to the vertical analogue ${\cal V}$ of
${\cal H}$. In section 5, splitting of the obtained field equations into its symmetric and skew-symmetric parts is performed. The symmetric part of the
resulting horizontal (resp. vertical)
field equations gives rise to a generalized form of
Einstein's field equations in which the horizontal (resp. vertical)
energy momentum tensor is purely geometric. The
skew-symmetric part of the resulting horizontal (resp. vertical) field equations gives rise to a generalized form of Maxwell's equations
in which the electromagnetic
field is purely geometric. In section 6, the form of the field equations under the Integrability, Cartan and Berwald conditions are deduced. These special cases
throw some light
on the role played by the nonlinear connection in the obtained field equations. Finally, in section 7, we end the paper by summarizing the
obtained results in the different cases dealt with and some concluding remarks.

\end{mysect}


\begin{mysect}{Basic Prelimenaries}

The material covered in the present section may be found in \cite{GLS} and \cite{V}. Some related topics may be found in \cite{GhA}.

\bigskip

Let $M$ be a paracompact manifold of dimension $n$ of class $C^{\infty}$. Let $\pi:TM\to M$ be its tangent bundle.
If $(U, \ x^{\mu})$
is a local chart on $M$, then $(\pi^{-1}(U), (x^{\mu}, \ y^{a}))$ is the corresponding local
chart on $TM$.  The coordinate transformation on $TM$ is given by:
\begin{equation}\label{chof}x^{\mu'} = x^{\mu'}(x^{\nu}), \ \ y^{a'} = p^{a'}_{a} y^{a},\end{equation}
$\mu = 1, \ldots, n; \ a = 1, \ldots, n$; \ 
 \begin{equation}\label{crucial}p^{a'}_{a} = \frac{\partial y^{a'}}{\partial y^{a}} =
\frac{\partial x^{a'}}{\partial x^{a}}\end{equation} and
det$(p^{a'}_{a})\neq 0$.

\bigskip

The tangent space $T_{u}(TM)$ at $u\in TM$ is a $2n$ dimensional vector space, having the natural basis
$(\frac{\partial}{\partial x_{\mu}}, \frac{\partial}{\partial y^{a}})$. The change of coordinates (\ref {chof}) {\bf in a local chart of $TM$} implies a
change of the natural basis as follows:

\begin{equation}\label{need}\frac{\partial}{\partial x^{\mu'}} = \frac{\partial x^{\mu}}{\partial x^{\mu'}}\frac{\partial}{\partial x^{\mu}}
+  \frac{\partial y^{a}}{\partial x^{\mu'}}\frac{\partial}{\partial y^{a}}\end{equation}

\begin{equation}\label{ned} \frac{\partial}{\partial y^{a'}} = \frac{\partial y^{a}}{\partial y^{a'}}\frac{\partial}{\partial y^{a}}\end{equation}

\vspace{0.1 cm}

The paracompactness of $M$ ensures the existence of a nonlinear connection $N$ on $TM$
with coefficients $N^{a}_{\alpha}(x, y)$. The transformation formula for the coefficients
$N^{a}_{\alpha}$ is given by
\begin{equation}\label{nlinear}N^{a'}_{\alpha'} = p^{a'}_{a} p^{\alpha}_{\alpha'}N^{a}_{\alpha} + p^{a'}_{a}
p^{a}_{c'\alpha'}y^{c'},\end{equation}
where $p^{a}_{c'\alpha'} = \frac{\partial p^{a}_{c'}}{\partial x^{\alpha'}}$ .
The nonlinear connection leads to the direct sum decomposition
\begin{equation}\label{Sum}T_{u}(TM) = H_{u}(TM)\oplus V_{u}(TM), \ \ \forall u\in
TM\setminus \{0\},\end{equation}
where $\{0\}$ is the null section of $TM$, $V_{u}(TM)$ is the vertical space at $u$ with
local basis $\dot \partial_a := \frac{\partial}{\partial y^{a}}$
and $H_{u}(TM)$ is the horizontal space at $u$, associated with $N$,
supplementary to $V_{u}(TM)$. If $\partial_{\mu}: = \frac{\partial}{\partial x^{\mu}}$, then
the canonical basis of $H_{u}(TM)$
is given by
\begin{equation}\delta_{\mu}: = \partial_{\mu} - N^{a}_{\mu} \ \dot {\partial_a}\end{equation}
Consequently, $(\delta_{\mu}, \, \dot{\partial_a})$ is a basis of $T_u(TM)$ at $u\in TM$, called the adapted basis.

\bigskip

Now, let $(dx^{\alpha}, \ \delta y^{a})$ be the basis of $T^{*}_{u}(TM)$ dual to the adapted
basis $(\delta_{\alpha}, \ \dot {\partial_a})$ of $T_{u}(TM)$. Then
\begin{equation}\delta y^{a} := dy^{a} + N^{a}_{\alpha}dx^{\alpha}\end{equation}
and
\begin{equation}dx^{\alpha}(\delta_{\beta}) = \delta^{\alpha}_{\beta}, \ \
dx^{\alpha}(\dot{\partial_{a}}) = 0; \ \ \ \delta y^{a}(\delta_{\beta}) = 0, \ \
\delta y^{a}(\dot \partial_{b}) = \delta^{a}_{b}.\end{equation}

Under a change of local coordinates in $TM$, the following holds:
\begin{equation}\label{change}\delta_{\alpha'} = p^{\alpha}_{\alpha'} \delta_{\alpha}, \ \ \dot{\partial_{a'}} =
p^{a}_{a'} \dot{\partial_{a}};
\ \ dx^{\alpha'} = p^{\alpha'}_{\alpha} dx^{\alpha}, \ \ \delta y^{a'} = p^{a'}_{a} \delta y^{a}.\end{equation}
The above transformation formulae result from the law of transformation (\ref{nlinear}) of the coefficients 
$(N^{a}_{\alpha})$ of the nonlinear connection, together with (\ref{need}) and (\ref{ned}).
The general covariance of geometric objects defined on the tangent bundle $TM$ is guaranteed by (\ref{change}) as will be revealed below.

\begin{definition} A nonlinear connection $N^{a}_{\mu}$ is said to be homogeneous if it is
positively homogeneous of degree $1$ in the directional argument $y$.
\end{definition}

We denote by $\mathfrak{X}(TM)$ the set of all vector fields on $TM$.

\begin {definition} A $d$-connection $D$ on $TM$ is a linear connection on $TM$ which preserves by
parallelism the horizontal and vertical distribution: if $Y$ is a horizontal (vertical) vector field,
then $D_{X} Y$ is a horizontal (vertical)
vector field, \  for all $X\in \mathfrak{X}(TM)$.
\end{definition}
Consequently, as opposed to a linear connection on $TM$, which has in general eight coefficients, a $d$-connection
$D$ on $TM$ has only four coefficients. The coefficients
of a $d$-connection $D = (\Gamma^{\alpha}_{\mu\nu}, \ \Gamma^{a}_{b\nu}, \ C^{\alpha}_{\mu c}, \ C^{a}_{bc})$
are defined by
\begin{equation}\label{xzc}D_{\delta \nu}\delta_{\mu} = :\Gamma^{\alpha}_{\mu\nu}\delta_{\alpha}, \ \ \
D_{\delta \nu}\dot{\partial_b} = :\Gamma^{a}_{b\nu}\dot {\partial_{a}}; \ \ \
D_{\dot{\partial_c}}\delta_{\mu} =: C^{\alpha}_{\mu c}\delta_{\alpha}, \ \ \
D_{\dot{\partial_c}}\dot {\partial_{b}} =: C^{a}_{bc}\dot{\partial_{a}}.\end{equation}

By (\ref{change}) and (\ref{xzc}), the transformation formulae of a $d$-connection are given by:
$$\Gamma^{\alpha'}_{\mu'\nu'} = p^{\alpha'}_{\alpha} p^{\mu}_{\mu'} p^{\nu}_{\nu'}
\Gamma^{\alpha}_{\mu\nu} + p^{\alpha'}_{\epsilon} p^{\epsilon}_{\mu'\nu'}, \ \ \ \Gamma^{a'}_{b'\mu'} = p^{a'}_{a} p^{b}_{b'} p^{\mu}_{\mu'}
\Gamma^{a}_{b\mu} + p^{a'}_{c}p^{c}_{b'\mu'};$$
$$C^{\alpha'}_{\mu' c'} = p^{\alpha'}_{\alpha} p^{\mu}_{\mu'} p^{c}_{c'}
C^{\alpha}_{\mu c}, \ \ \ C^{a'}_{b'c'} = p^{a'}_{a} p^{b}_{b'} p^{c}_{c'}
C^{a}_{bc}.$$
\begin{definition} A $d$-tensor field $A$ on $TM$ of type $(p, r; \ q, s)$ is a tensor field on
$TM$ which can be locally expressed in the form
$$A = A^{{u_{1}}...{u_{p + r}}}_{{v_1}\ldots{v_{q+s}}}\partial_{u_{1}}\otimes\ldots
\otimes{\partial_{u_{p + r}}}\otimes dx^{v_1}\otimes\ldots\otimes dx^{v_{q + s}},$$
where $u_{i}\in \{\alpha_{i}, \ a_i\}, \ v_{j}\in \{\beta_{j}, \ b_{j}\}$,
$$ \ \partial_{u_{i}}\in \{\delta_{\alpha_{i}}, \ \dot{\partial}_{a_i}\}, \
dx^{v_{j}}\in \{dx^{\beta_j}, \ \delta y^{b_j}\}, \ \ i = 1,\ldots, p + r; \ j = 1,\ldots, q + s,$$
so that the number of $\alpha_{i}$'s = $p$, the number of $a_{i}$'s = $r$, the number of
$\beta_{j}$'s = $q$ and the number of $b_{j}$'s = $s$.
\end{definition}

Let $A = A^{\alpha a}_{\beta b}\delta_{\alpha}\otimes{\dot{\partial_a}}\otimes dx^{\beta}
\otimes \delta y^{b}$ be a $d$-tensor field of type $(1, 1; 1, 1)$.
If $$A = A^{\alpha' a'}_{\beta' b'}\delta_{\alpha'}\otimes{\dot{\partial_{a'}}}\otimes dx^{\beta'}
\otimes \delta y^{b'}$$ is the representation of $A$ in the new coordinate system $(x^{\mu'}, y^{a'})$, then, in view of (\ref{change}), we have
\begin{equation}  A^{\alpha' a'}_{\beta' b'} = p^{\alpha'}_{\alpha} p^{\beta}_{\beta'} p^{a'}_{a} p^{b}_{b'}A^{\alpha a}_{\beta b}.\end{equation}

Since both $p^{\alpha}_{\alpha'}$ and $p^{a}_{a'}$ depend on the positional argument $x$ only, as is clear from (\ref{crucial}), it follows that
{\it the transformation formula for the components of a $d$-tensor field on the tangent bundle $TM$ is similar in form to the transformation formula for the
components of a tensor field defined on the base manifold $M$.} In fact, this should be expected; a $d$-tensor field is defined in terms of the adapted basis {\it and not in terms of the natural basis}. The adapted basis, unlike the natural basis, has a transformation formula similar in form to the transformation formula of a tensor field defined on the base manifold $M$. For this reason, \emph{the general covariance of the geometric objects considered in this work is ensured}.

\bigskip

{\bf A comment on notation:} Throughout the paper, both Greek indices $\{\alpha, \beta, \mu,\ldots\}$ and
Latin indices $\{a, b, c,\ldots\}$ take the values $1,\ldots, n$. {\bf Greek} indices are used to denote
{\bf horizontal} entities, whereas {\bf Latin} indices are used to denote {\bf vertical}
 entities. Einstein convention is applied on both types of indices. Also, throughout, the symbol $|$ and $||$ will
denote the horizontal and vertical covariant derivatives respectively with respect to a given $d$-connection.

\begin{definition} The torsion tensor field $\bf{T}$ of a d-connection $D$ on $TM$ is defined by
\begin{equation}\label{def}{\bf T}(X, Y): = D_{X} Y - D_{Y} X - [X, Y]; \ \ \forall X, Y\in \mathfrak{X}(TM).\nonumber\end{equation}
\end{definition}

\begin{proposition} In the adapted basis $(\delta_{\alpha}, \ \dot{\partial_{a}})$,
the torsion tensor field ${\bf T}$ of a d-connection $D =
(\Gamma^{\alpha}_{\mu\nu}, \, \Gamma^{a}_{b \mu}, \, C^{\alpha}_{\mu
c}, \, C^{a}_{bc})$ is characterized by the $d$-tensor fields with
local coefficients $(\Lambda^{\alpha}_{\mu\nu}, \, R^{a}_{\mu\nu},
\, C^{\alpha}_{\mu c}, \, P^{a}_{\mu c}, \, T^{a}_{bc})$ defined by:
$$h{\bf T}(\delta_{\nu}, \ \delta_{\mu}) =: \Lambda^{\alpha}_{\mu\nu}\delta_{\alpha}, \ \ \ \ \
v{\bf T}(\delta_{\nu}, \ \delta_{\mu}) =:
R^{a}_{\mu\nu}\dot{\partial_a}$$
$$h{\bf T}(\dot{\partial_c}, \ \delta_{\mu}) =: C^{\alpha}_{\mu c}\delta_{\alpha}, \ \ \ \ \
v{\bf T}(\dot{\partial_c}, \ \delta_{\mu}) =: P^{a}_{\mu c}\dot{\partial_a}, \ \ \ \ \ \ v{\bf T}(\dot{\partial_c}, \
\dot{\partial_b}) =: T^{a}_{bc}\dot{\partial_a},$$
\begin{equation}\Lambda^{\alpha}_{\mu\nu} := \Gamma^{\alpha}_{\mu\nu} - \Gamma^{\alpha}_{\nu\mu}, \ \ \
R^{a}_{\mu\nu} := \delta_{\nu}N^{a}_{\mu} - \delta_{\mu}N^{a}_{\nu}, \ \ \ P^{a}_{\mu c}: = \dot{\partial_c} N^{a}_{\mu} - \Gamma^{a}_{c \mu}, \ \ \
T^{a}_{bc} := C^{a}_{bc} - C^{a}_{cb},\nonumber\end{equation}
where $h$ (resp. $v$) denotes the horizontal (resp. vertical) counterpart.

\end{proposition}

It is to be noted that $R^{a}_{\mu\nu}$ is again the curvature of the nonlinear connection $N$.

\begin{definition}The curvature tensor field ${\bf R}$ of a $d$-connection $D$ is given by
\begin{equation}{\bf R}(X, Y)Z: = D_{X} D_{Y} Z - D_{Y}D_{X} Z - D_{[X, \ Y]}Z; \ \ \forall X, Y, Z\in \mathfrak{X}(TM).\nonumber\end{equation}
\end{definition}

In the adapted basis $(\delta_{\alpha}, \,\dot{\partial_a})$, the curvature tensor field ${\bf R}$ is characterized by the six $d$-tensor fields defined by
$$ \ {\bf R}(\delta_{\mu}, \ \delta_{\nu})\delta_{\beta} =: R^{\alpha}_{\beta\mu\nu}\delta_{\alpha};
\ \ \ \ \ {\bf R}(\delta_{\mu}, \ \delta_{\nu})\dot{\partial_{b}} =: R^{a}_{b\mu\nu}\dot{\partial_{a}},$$
$${\bf R}(\dot{\partial_c}, \ \delta_{\nu})\delta_{\beta} =: P^{\alpha}_{\beta\nu c}\delta_{\alpha};
\ \ \ \ \ \ {\bf R}(\dot{\partial_c}, \ \delta_{\nu})\dot{\partial_{b}} =: P^{a}_{b\nu c}\dot{\partial_{a}},$$
$${\bf R}(\dot{\partial_b}, \ \dot{\partial_c})\delta_{\beta} =: S^{\alpha}_{\beta bc}\delta_{\alpha};
\ \ \ \ \ \ {\bf R}(\dot{\partial_c}, \ \dot{\partial_d})\dot{\partial_{b}} =: S^{a}_{bcd}\dot{\partial_{a}}.$$

\begin{definition} An $hv$-metric on $TM$ is a covariant $d$-tensor field
${\bf {\cal G}} := h{\bf {\cal G}} + v{\bf {\cal G}}$ on $TM$, where $h{\bf {\cal G}}: = g_{\alpha\beta}\,dx^{\alpha}\otimes dx^{\beta}$, \
$v{\bf {\cal G}} := g_{ab}\,\delta y^{a}\otimes \delta y^{b}$ such that:
\begin{equation}g_{\alpha\beta} = g_{\beta\alpha}, \ \ det (g_{\alpha\beta})\neq 0; \ \
\ g_{ab} = g_{ba}, \ \ det (g_{ab})\neq 0.\nonumber\end{equation}
\end{definition}
\begin{definition} A $d$-connection $D$ on $TM$ is said to be metric or compatible with
the metric ${\bf {\cal G}}$ if $D_{X}{\bf {\cal G}} = 0, \ \forall X\in \mathfrak{X}(TM)$.
\end{definition}

In the adapted frame $(\delta_{\alpha}, \ \dot{\partial_a})$, the above condition can
be expressed in the form:
\begin{equation}\label{locally}g_{\alpha\beta|\mu} = g_{\alpha\beta||c} = g_{ab|\mu} = g_{ab||c} = 0.\nonumber\end{equation}

\begin{theorem}\label{metric} For a given $hv$-metric on $TM$, there exists a unique metric $d$-connection \ $\overcirc{D} =
( \ \overcirc{\Gamma^{\alpha}_{\mu\nu}}, \ \, \overcirc{\Gamma^{a}_{b\nu}}, \
\, \overcirc{C^{\alpha}_{\mu c}}, \ \, \overcirc{C^{a}_{bc}})$ on $TM$ with the properties that
\begin{description}
\item [(a)] $\overcirc{\Lambda}^{\alpha}_{\mu\nu} = \ \overcirc{\Gamma}^{\alpha}_{\mu\nu} - \
\overcirc{\Gamma}^{\alpha}_{\nu\mu} = 0, \ \, \ \ \ \ \overcirc{T}^{a}_{bc} = \ \overcirc{C}^{a}_{bc} - \ \overcirc{C}^{a}_{cb} = 0.$
\item [(b)] $\overcirc{\Gamma}^{a}_{b\nu} := \dot{\partial_b} N^{a}_{\nu} +
\frac{1}{2} \, g^{ac}(\delta_{\nu} g_{bc} - g_{dc} \, \dot{\partial_b}N^{d}_{\nu} -
g_{bd} \, \dot{\partial_c}N^{d}_{\nu}), \ \ \ \ \ \ \, \overcirc{C}^{\alpha}_{\mu c} := \frac{1}{2} \, g^{\alpha\epsilon}\dot{\partial_c}g_{\mu\epsilon}.$
\end{description}
In this case, the coefficients \ $\overcirc{\Gamma}^{\alpha}_{\mu\nu}$ and \ $\overcirc{C}^{a}_{bc}$ are necessarily
of the form
\begin{equation}\overcirc{\Gamma}^{\alpha}_{\mu\nu}: = \frac{1}{2} \, g^{\alpha\epsilon}(\delta_{\mu}
g_{\epsilon\nu} + \delta_{\nu}g_{\epsilon\mu} - \delta_{\epsilon}g_{\mu\nu}), \ \ \ \ \
\overcirc{C^{a}_{bc}}: = \frac{1}{2} \, g^{ad}(\dot{\partial_{b}}g_{dc} +
\dot{\partial_{c}}g_{db} - \dot{\partial_{d}}g_{bc}).\nonumber\end{equation}

\end{theorem}
We call the connection \,$\overcirc{D}$ the {\bf natural metric} $d$-connection.

\begin{definition} A nonlinear connection $N^{a}_{\mu}$ is said to be integrable if $R^{a}_{\mu\nu} = 0.$ In this case, the bracket of two horizontal
vector fields is a horizontal vector field.
\end{definition}

\begin{proposition} \label{SEE}Let \ $\overcirc{R}_{\beta\nu}: = \ \overcirc{R}^{\alpha}_{\beta\nu\alpha}, \ \ \overcirc{S}_{bc} :=  \ \overcirc{S}^{d}_{bcd}$
be the horizontal and vertical Ricci tensors respectively of the above metric $d$-connection. Then we have:

\begin{equation}\label{skewz}\overcirc{R}_{\beta\nu} - \ \overcirc{R}_{\nu\beta} = \, \mathfrak{S}_{\beta, \nu, \alpha} \
\overcirc{C}^{\alpha}_{\beta d}R^{d}_{\nu\alpha}, \ \ \ \ \ \overcirc{S}_{bc} = \ \overcirc{S}_{cb},\end{equation}
where $\mathfrak{S}_{\beta, \nu, \alpha}$ denotes a cyclic sum on the indices $\beta, \nu, \alpha$.
Consequently, \ $\overcirc{R}_{\beta\nu}$ is symmetric if the nonlinear connection
is integrable.
\end{proposition}

\begin{theorem}\label{EE} Let
$\,\,\overcirc{\cal R} := g^{\alpha\beta}\,\,\overcirc{R}_{\alpha\beta}$ and $\,\,\overcirc{\cal S}: = h^{ab}\,\,\overcirc{S}_{ab}$ be the
horizontal and vertical Ricci scalars
of the natural metric $d$-connection respectively. Then the horizontal Einstein tensor
\begin{equation}\label{HE}\overcirc{G}_{\mu\sigma} := \ \overcirc{R}_{\mu\sigma} - \frac{1}{2} \, g_{\mu\sigma} \ \overcirc{\cal{R}}
\end{equation}
satisfies the identity
\begin{equation}\label{HES} \,\overcirc{G}^{\mu}\!\,_{\sigma{o\atop|}\mu} = R^{a}_{\sigma\mu}\,\,\overcirc{P}^{\mu}_{a} +
\frac{1}{2}\,R^{a}_{\alpha\mu}\,\,\overcirc{P}^{\alpha\mu}\!\,_{\sigma a},\end{equation}
where $\,\overcirc{G}^{\mu}_{\sigma}:= g^{\mu\epsilon}\,\overcirc{G}_{\epsilon\sigma}$, $\,\,\overcirc{P}_{\epsilon a} := - \,\,\overcirc{P}^{\beta}_{\epsilon\beta a}$,
$\,\,\overcirc{P}^{\mu}_{a} := g^{\mu\epsilon}\,\overcirc{P}_{\epsilon a}$, $\,\,\overcirc{P}^{\alpha\mu}\!\,_{\sigma a} := g^{\mu\epsilon}\,\overcirc{P}^{\alpha}_{\epsilon\sigma a}$.

Consequently, if the nonlinear connection is integrable, then $\,\overcirc{G}_{\mu\sigma}$ is
symmetric and satisfies the conservation law
\begin{equation}\label{hee}\,\,\overcirc{G}^{\mu}\!\,_{\sigma{o\atop|}\mu} = 0.\end{equation}
\par On the other hand, the vertical Einstein tensor  \
\begin{equation}\overcirc{G}_{ab} := \ \overcirc{S}_{ab} - \frac{1}{2} \, h_{ab} \ \overcirc{\cal{S}}\end{equation} is symmetric and satisfies the
conservation law
\begin{equation}\label{vee}\,\,\overcirc{G}^{a}\!\,_{b{o\atop||}a} = 0,\end{equation}
where \,$\overcirc{G}^{a}_{b}: = g^{ac}\,\,\overcirc{G}_{cb}.$
\end{theorem}

\end{mysect}


\begin{mysect}{Extended Absolute Parallelism Geometry (EAP-geometry)}

Searching for a viable modern formulation of conventional AP-geometry may seem not only desirable, but actually
essential. The structure of EAP-geometry may reveal some deep connection between the geometry of the tangent bundle, which is geometrically very rich, and the
conventional AP-geometry, which is physically quite successful. The inherent simplicity of EAP-geometry may indicate that such a connection is neither artificial
nor can be overlooked.

\bigskip

In this section, we give a brief review of the basic concepts of the EAP-geometry.
We shall limit ourselves to the necessary material needed for the construction of the field equations.
For a detailed exposition of EAP-geometry, we refer the reader to \cite{EAP}.

\bigskip

As in the previous section, $M$ is assumed to be a smooth paracompact manifold of dimension $n$. This insures the existence of a
nonlinear connection on
$TM$ so that the decomposition (\ref{Sum}) induced by the nonlinear connection holds.

\bigskip

We assume that \ $\undersym{\lambda}{i}$, $i = 1, \ldots, n$, are $n$ vector fields {\bf globally} defined
on $TM$. In the adapted basis $(\delta_{\alpha}, \ \dot{\partial_{a}})$, we have \
$\undersym{\lambda}{i} = \ h \, \undersym{\lambda}{i} + \ v \, \undersym{\lambda}{i} =
 \ \undersym{\lambda}{i}^{\alpha}{\delta_{\alpha}} +
\ \undersym{\lambda}{i}^{a}\dot{\partial_a}.$ We further assume that the $n$ horizontal vector fields $h \, \undersym{\lambda}{i}$ and
the $n$ vertical vector fields $v \, \undersym{\lambda}{i}$ are {\bf linearly
independent} so that
\begin{equation}\label{inverse}\undersym{\lambda}{i}^{\alpha} \ \undersym{\lambda}{i}_{\beta} = \delta^{\alpha}_{\beta}, \
\ \ \ \undersym{\lambda}{i}^{\alpha} \ \undersym{\lambda}{j}_{\alpha} = \delta_{ij}; \ \ \ \
\ \undersym{\lambda}{i}^{a} \ \undersym{\lambda}{i}_{b} = \delta^{a}_{b}, \
\ \ \ \undersym{\lambda}{i}^{a} \ \undersym{\lambda}{j}_{a} = \delta_{ij},\nonumber\end{equation}
where $( \ \undersym{\lambda}{i}_{\alpha})$ and $( \ \undersym{\lambda}{i}_{a})$ denote the
inverse matrices of $( \ \undersym{\lambda}{i}^{\alpha})$ and $( \ \undersym{\lambda}{i}^{a})$
respectively. We refer to the vector fields ($\,\undersym{\lambda}{i}^{\alpha}$ and $\,\undersym{\lambda}{i}^{a}$) $\lambda$'s as
the (horizontal and vertical) {\bf fundamental vector fields}. These $\,\,\undersym{\lambda}{i}$
correspond to the ordinary tetrad of the conventional AP-space. It should be noted that the fundamental vector fields depend on both the positional
argument $x$ and the directional argument $y$. Consequently, all geometric objects of the EAP-space are a function of both $x$ and $y$.

\bigskip

A manifold $M$ equipped with such independent vector fields on $TM$ is called an Extended Absolute Parallelism space (EAP-space) and is
denoted by $(TM, \lambda)$.

\begin{remark}\em{It should be noted that the existence of mesh indices on any geometric object does not in any way affect the
general covariance of the tensor fields involved. For example, the transformation formula
for $\undersym{\lambda}{i}^{\alpha}$ (a tensor field of type $(1, 0; 0, 0)$) is given by \,$\undersym{\lambda}{i}^{\alpha'} = p^{\alpha'}_{\alpha}\,
\undersym{\lambda}{i}^{\alpha}$. Moreover, if \, $\undersym{C}{i} = \,\undersym{\lambda}{i}^{\mu}C_{\mu}$, then, noting that 
$C_{\mu}$ transforms according to the law $C_{\mu'} = p^{\mu}_{\mu'}C_{\mu}$, the $\,\undersym{C}{i}$'s,
$i = 1,\ldots, n$, {\it transform as a set of scalars}. Indeed, we have
\\[- 0.2 cm]$$\undersym{C'}{i} = \,\undersym{\lambda}{i}^{\mu'} C_{\mu'} = p^{\mu'}_{\mu} p^{\alpha}_{\mu'}\,\undersym{\lambda}{i}^{\mu}C_{\alpha} =
\delta^{\alpha}_{\mu} \,\,\undersym{\lambda}{i}^{\mu}C_{\alpha} =
 \,\undersym{\lambda}{i}^{\mu} C_{\mu} = \,\undersym{C}{i}.$$
On the other hand, if a summation is performed on the mesh index $i$, then in this case $i$ acts as a dummy index and hence does not again affect
the covariance of the tensor fields involved.}
\end{remark}

We will use the symbol $\lambda$ without
the subscript $i$ to denote any one of the vector fields \
$\undersym{\lambda}{i} \ (i = 1, \ldots, n)$. The index $i$
will appear {\it only when summation is performed.}

\bigskip

We set
\begin{equation}\label{apm}g_{\alpha\beta}: = \ \undersym{\lambda}{i}_{\alpha} \ \undersym{\lambda}{i}_{\beta}, \ \
\ \ g_{ab}: = \ \undersym{\lambda}{i}_{a} \ \undersym{\lambda}{i}_{b}.\nonumber\end{equation} Then, clearly,
$${\bf {\cal G}} =  g_{\alpha\beta} \,dx^{\alpha}\otimes dx^{\beta} + g_{ab}\,\delta y^{a}\otimes \delta y^{b}$$
is an $hv$-metric on $TM$. Moreover, the inverse of the matrices $(g_{\alpha\beta})$ and $(g_{ab})$
are given by $(g^{\alpha\beta})$ and $(g^{ab})$ respectively, where
\begin{equation}g^{\alpha\beta} = \ \undersym{\lambda}{i}^{\alpha} \
\undersym{\lambda}{i}^{\beta}, \ \ \ \ g^{ab} = \ \undersym{\lambda}{i}^{a} \
\undersym{\lambda}{i}^{b}.\nonumber\end{equation}

\begin{theorem}\label{apc} The $d$-connection $D = (\Gamma^{\alpha}_{\mu\nu}, \
\Gamma^{a}_{b\nu}, \ C^{\alpha}_{\mu c}, \ C^{a}_{bc})$ defined by
\begin{equation}\Gamma^{\alpha}_{\mu\nu} =
\ \undersym{\lambda}{i}^{\alpha}(\delta_{\nu} \ \undersym{\lambda}{i}_{\mu}),  \ \ \
\Gamma^{a}_{b\nu} = \ \undersym{\lambda}{i}^{a}(\delta_{\nu} \ \undersym{\lambda}{i}_{b});  \ \
\ C^{\alpha}_{\mu c} = \ \undersym{\lambda}{i}^{\alpha} (\dot{\partial_c} \
\undersym{\lambda}{i}_{\mu}), \ \ \ C^{a}_{bc} = \ \undersym{\lambda}{i}^{a} (\dot{\partial_c} \
\undersym{\lambda}{i}_{b})\nonumber\end{equation}
satisfies the AP-condition
\begin{equation}\label{NY} {\lambda}^{\alpha}\!\, _{|\mu} = {\lambda}^{\alpha}\!\, _{||c} =
{\lambda}^{a}\!\, _{|\mu} = {\lambda}^{a}\!\, _{||c} = 0.\nonumber\end{equation}
Consequently, $D$ is a metric $d$-connection.
\end{theorem}

This $d$-connection is referred to as the {\bf canonical} $d$-connection.

\begin{definition}The torsion tensor field ${\bf T} = (\Lambda^{\alpha}_{\mu\, \nu}, \,R^{a}_{\mu\nu},
\,C^{\alpha}_{\mu c}, \,P^{a}_{\mu c}, \,T^{a}_{bc})$ of the canonical $d$-connection is referred to
as the torsion of the EAP-space.
\end{definition}

\begin{definition} The contortion tensor field of the EAP-space is defined by
\begin{equation}{\bf C}(X, \ Y): = D_{Y}X - \ \overcirc{D}_{Y}X; \ \ \forall X, Y\in \mathfrak{X}(TM),\nonumber\end{equation}
where $D_{Y}$ is  the covariant derivatives with respect to the cannonical $d$-connection
and \, $\overcirc{D}_{Y}$ is the covariant derivatives with respect to the natural metric
$d$-connection obtained in Theorem \ref{metric} corresponding to the metric defined by (\ref {apm}).
\end{definition}

In the adapted basis $(\delta_{\mu}, \ \dot{\partial_a})$, the contortion tensor ${\bf C}$ is characterized by the $d$-tensor fields
with local coefficients $(\gamma^{\alpha}_{\mu\nu}, \, \gamma^{a}_{b\nu}, \, \gamma^{\alpha}_{\mu c}, \, \gamma^{a}_{bc})$ defined by:
\begin{equation}\label{cnt}\gamma^{\alpha}_{\mu\nu} := \Gamma^{\alpha}_{\mu\nu} -
\ \overcirc{\Gamma}^{\alpha}_{\mu\nu}, \ \ \ \gamma^{a}_{b\mu}: =
\Gamma^{a}_{b\mu} -
\ \overcirc{\Gamma}^{a}_{b\mu}; \ \ \ \ \gamma^{\alpha}_{\mu c}: = C^{\alpha}_{\mu c} -
\ \overcirc{C}^{\alpha}_{\mu c}, \ \ \ \gamma^{a}_{bc}: = C^{a}_{bc} - \ \overcirc{C}^{a}_{bc}.\end{equation}
Consequently,
\begin{equation}\label{torsion}\Lambda^{\alpha}_{\mu\nu} =  \gamma^{\alpha}_{\mu\nu} - \gamma^{\alpha}_{\nu\mu}, \ \ \
\ \ \ \ \ T^{a}_{bc} = \gamma^{a}_{bc} - \gamma^{a}_{cb}.\end{equation}
We set
\begin{equation}\Omega^{\alpha}_{\mu\nu} := \gamma^{\alpha}_{\mu\nu} + \gamma^{\alpha}_{\nu\mu}, \ \ \ \ \ \ \ \ \Omega^{a}_{bc} :=
\gamma^{a}_{bc} + \gamma^{a}_{cb}.\nonumber\end{equation}

By definition of the canonical $d$-connection and (\ref{cnt}), we have

\begin{proposition} The contortion tensor can be expressed in terms of
the ${\lambda}$'s in the form:
\begin{equation}\label{contortion}\gamma^{\alpha}_{\mu\nu} = \ \undersym{\lambda}{i}^{\alpha} \
\undersym{\lambda}{i}_{\mu{o\atop{|}}\nu}, \ \
\gamma^{a}_{b \mu} = \ \undersym{\lambda}{i}^{a} \
\undersym{\lambda}{i}_{b {o\atop{|}}\mu}, \ \ \gamma^{\alpha}_{\mu c} =
\ \undersym{\lambda}{i}^{\alpha} \
\undersym{\lambda}{i}_{\mu{o\atop{||}}c}, \ \ \ \gamma^{a}_{b c} =
\ \undersym{\lambda}{i}^{a} \ \undersym{\lambda}{i}_{b {o\atop{||}}c}.\nonumber\end{equation}
\end{proposition}

\begin{proposition}\label{skew} Let $\gamma_{\alpha\mu\nu}: =
g_{\alpha\epsilon}\gamma^{\epsilon}_{\mu\nu}, \
\gamma_{ab\mu}: = g_{ac}\gamma^{c}_{b \mu}, \ \gamma_{\alpha\mu c} :=
g_{\alpha\epsilon}\gamma^{\epsilon}_{\mu c}$, \ $\gamma_{abc} := g_{ad}\gamma^{d}_{bc}$.
Then each of the above defined $d$-tensor fields is skew-symmetric in the
first pair of indices.
Consequently,
$\gamma^{\alpha}_{\alpha \nu} = \gamma^{a}_{a\mu} = \gamma^{\alpha}_{\alpha c} =
\gamma^{a}_{a c} = 0$.
\end{proposition}

We set
\begin{equation}\Lambda^{\alpha}_{\mu\alpha} = \gamma^{\alpha}_{\mu\alpha} = :C_{\mu},  \ \ \ \ \ \
T^{a}_{ba} = \gamma^{a}_{ba} = :C_{b}.\nonumber\end{equation}
\begin{definition}We refer to ${\bf B}: = (C_{\mu}, \, C_{a})$ as the basic vector field of the EAP-space.
\end{definition}
\begin{proposition}\label{Bianchi} The following identities hold:

\begin{equation}\label{h}\Lambda^{\alpha}_{\mu\nu|\alpha} =
(C_{\mu|\nu} - C_{\nu|\mu}) + C_{\epsilon}\Lambda^{\epsilon}_{\mu\nu} + \
\mathfrak{S}_{\mu, \nu, \alpha} C^{\alpha}_{\nu a}R^{a}_{\mu\alpha}\end{equation}

\begin{equation}\label{v}T^{d}_{bc||d} = (C_{b||c} - C_{c||b}) + C_{d}T^{d}_{bc}.\end{equation}

\end{proposition}

Now, we list in Table 1 below some second rank tensors available in
the EAP-space. These tensors will play an important role in the
constructed field equations.

\newpage
\begin{center}{\bf Table 1: Fundamental second rank tensors of EAP-space}\\[0.1 cm]
{\begin{tabular}{|c|c|c|c|}\hline
\multicolumn{2}{|c|}{\hbox{ }} &\multicolumn{2}{c|}{\hbox{ }}\\
\multicolumn{2}{|c|}{{\bf Horizontal}}&\multicolumn{2}{c|}{{\bf Vertical}}
\\[0.4 cm]\hline
&&&\\
{\bf Skew-Symmetric}&{\bf Symmetric}&{\bf Skew-Symmetric}&{\bf Symmetric}
\\[0.4 cm]\hline
&&&\\
$\xi_{\mu\nu}: = \gamma_{\mu\nu}  \!^{\alpha} \!\,_{|\alpha}$&$ \ $&
${\xi}_{ab}: = \gamma_{ab}  \!^{d} \!\,_{||d}$&$ \ $
\\[0.4 cm]\hline
&&&\\
$\gamma_{\mu\nu}: = C_{\alpha}\gamma_{\mu\nu} \!^{\alpha}$& $ \
$&${\gamma}_{ab} : = C_{d}\gamma_{ab} \!^{d}$& $ \ $
\\[0.4 cm]\hline
&&&\\
$\eta_{\mu\nu} := C_{\beta}\,\Lambda^{\beta}_{\mu\nu}$&
$\phi_{\mu\nu} := C_{\beta}\,\Omega^{\beta}_{\mu\nu}$ &$
{\eta}_{ab} := C_{d}\,T^{d}_{ab}$& $
{\phi}_{ab} := C_{d}\,\Omega^{d}_{ab}$
\\[0.4 cm]\hline
&&&\\
$\chi_{\mu\nu} := \Lambda^{\alpha}_{\mu\nu|\alpha}$& $\psi_{\mu\nu}
:= \Omega^{\beta}_{\mu\nu|\beta}$& $ {\chi}_{ab} :=
T^{d}_{ab||d}$& $ {\psi}_{ab} :=
\Omega^{d}_{ab||d}$
\\[0.4 cm]\hline
&&&\\
$\epsilon_{\mu\nu} := C_{\mu|\nu} -
C_{\nu|\mu}$& $\theta_{\mu\nu} :=  C_{\mu|\nu}
+ C_{\nu|\mu}$& ${\epsilon}_{ab} :=
C_{a||b} - C_{b||a}$& ${\theta}_{ab} :=
C_{a||b} + C_{b||a}$
\\[0.4 cm]\hline
&&&\\
\tiny{$\kappa_{\mu\nu} :=
\gamma^{\beta}_{\alpha\mu}\gamma^{\alpha}_{\nu\beta}
 - \gamma^{\beta}_{\mu\alpha}\gamma^{\alpha}_{\beta\nu}$}&
\tiny{$\varpi_{\mu\nu}: =
\gamma^{\beta}_{\alpha\mu}\gamma^{\alpha}_{\nu\beta} +
\gamma^{\beta}_{\mu\alpha}\gamma^{\alpha}_{\beta\nu}$}&

\tiny{${\kappa}_{ab}: =
\gamma^{c}_{da}\gamma^{d}_{bc} -
\gamma^{c}_{ad}\gamma^{d}_{cb}$}&
\tiny{${\varpi}_{ab} :=
\gamma^{c}_{da}\gamma^{d}_{bc} +
\gamma^{c}_{ad}\gamma^{d}_{cb}$}
\\[0.4 cm]\hline
&&&\\
$ \ $&$\sigma_{\mu\nu} :=
\gamma^{\beta}_{\alpha\mu}\gamma^{\alpha}_{\beta\nu}$& $ \ $&
${\sigma}_{ab}: =
\gamma^{c}_{da}\gamma^{d}_{cb}$

\\[0.4 cm]\hline
&&&\\
$ \ $&$\omega_{\mu\nu} :=
\gamma^{\beta}_{\mu\alpha}\gamma^{\alpha}_{\nu\beta}$& $  \ $&
${\omega}_{ab} :=
\gamma^{c}_{ad}\gamma^{d}_{bc}$
\\[0.4 cm]\hline
&&&\\
$ \ $&$\alpha_{\mu\nu} := C_{\mu}C_{\nu}$& $ \ $& $
{\alpha}_{ab} := C_{a}C_{b}$
\\[0.4 cm]\hline
\end{tabular}}
\end{center}

\begin{proposition}\label{RICCIT} The Ricci tensors $\,\,\overcirc{R}_{\beta\mu}$ and
$\,\,\overcirc{S}_{bc}$ of the natural metric $d$-connection can be expressed in terms of the fundamental tensors of Table 1 in the form
\begin{equation}\label{symx}\overcirc{R}_{\beta\mu} = - \frac{1}{2}\,(\theta_{\beta\mu} - \psi_{\beta\mu} + \phi_{\beta\mu}) + \omega_{\beta\mu} +
Q_{(\beta\mu)} + \frac{1}{2}\, \mathfrak{S}_{\beta, \mu, \alpha} \ \overcirc{C}^{\alpha}_{\beta a} R^{a}_{\mu \alpha},\\[- 0.2 cm]\end{equation}
\begin{equation}\label{VSM}\,\overcirc{S}_{bc} =  - \frac{1}{2}\,(\theta_{bc} - \psi_{bc} + \phi_{bc}) + \omega_{bc};\end{equation}
where $Q_{\beta\mu} := \gamma^{\epsilon}_{\beta a}R^{a}_{\mu\epsilon}$. 
\end{proposition}

\end{mysect}

\newpage


\begin{mysect}{Some Useful Relations}
Henceafter, we shall make extensive use of the fundamental tensor fields defined in Table 1 of the previous section.
We shall denote $\frac {\partial }{\partial x^{\mu}}\equiv \partial_{\mu}$ by \lq\lq$,\mu$\rq\rq \,and
$\frac {\partial }{\partial y^{a}}\equiv \dot{\partial_{a}}$ by \lq\lq$; a$\rq\rq \, interchangeably.
Moreover, {\it we regard $\lambda_{\beta}$, $\lambda_{\beta, \alpha}$, $\lambda_{\beta; a}$, $\lambda_{b}$, $\lambda_{b, \alpha}$ and
$\lambda_{b; a}$ as independent}. \footnote{Let $A = \{\lambda_{\beta}, \lambda_{\beta, \alpha}, \lambda_{\beta; a}, \lambda_{b},
\lambda_{b, \alpha}, \lambda_{b; a}\}$. By saying that $u, v\in A$ are independent we mean that
$\frac{\partial u}{\partial v} = \frac{\partial v}{\partial u} = 0$.} The next three Lemmas are needed for the derivation of the field equations. The proof of the
first two Lemmas is not difficult and will be omitted.

\begin{lemma}\label{ct}
Let $C := C^{\epsilon}\,_{|\epsilon} - C^{\epsilon}C_{\epsilon}, \,\theta : = g^{\mu\nu}\theta_{\mu\nu}$,
$\alpha := g^{\mu\nu}\alpha_{\mu\nu}$, $\phi: = g^{\mu\nu}\phi_{\mu\nu}$, $\psi := g^{\mu\nu}\psi_{\mu\nu}$ and
$\Lambda_{\beta\mu\nu} := g_{\beta\epsilon}\Lambda^{\epsilon}_{\mu\nu}$. Then the following relations hold:
\begin{description}
\item [(a)] $\frac{1}{2}\, (\phi - \psi + \theta) - \alpha - 2C = 0$,
\item [(b)] $C^{\epsilon}(\Lambda_{\mu\epsilon\nu} + \Lambda_{\nu\epsilon\mu}) = - \,\phi_{\mu\nu}, \ \ \ \
g^{\epsilon\alpha}(\Lambda_{\mu\nu\alpha|\epsilon} + \Lambda_{\nu\mu\alpha|\epsilon}) = \psi_{\mu\nu},$
\item [(c)] $C^{\epsilon}(\Lambda_{\mu\epsilon\nu} - \Lambda_{\nu\epsilon\mu}) = - \,(2\gamma_{\mu\nu} + \eta_{\mu\nu}), \ \ \ \
g^{\epsilon\alpha}(\Lambda_{\mu\nu\alpha|\epsilon} - \Lambda_{\nu\mu\alpha|\epsilon}) = 2 \xi_{\mu\nu} + \chi_{\mu\nu}.$
\end{description}
\end{lemma}

\begin{lemma}\label{lemma} Let $(TM, \lambda)$ be an EAP-space. Let $D = (\Gamma^{\alpha}_{\mu\nu}, \, \Gamma^{a}_{b\mu}, \, C^{\alpha}
_{\mu c}, \, C^{a}_{bc})$ be the canonical $d$-connection. Then the following hold:
\begin{description}
\item [(a)] $\frac{\partial |\lambda|}{\partial \,\undersym{\lambda}{j}_{\beta}} = |\lambda|\,\undersym{\lambda}{j}^{\beta},$ \
$|\lambda|: = det(\lambda_{\beta})$.
Consequently, $|\lambda|\!\,_{,\gamma} = |\lambda|(\,\undersym{\lambda}{j}^{\beta}\,\undersym{\lambda}{j}_{\beta,\gamma})$, \ \
$\delta_{\gamma}|\lambda| = |\lambda|\Gamma^{\beta}_{\beta\gamma}$, \ \ $|\lambda|\!\,_{; a} = |\lambda|C^{\beta}_{\beta a}.$

\item [(b)] $\frac{\partial \,\undersym{\lambda}{i}^{\alpha}}{\partial \,\undersym{\lambda}{j}_{\beta}} = - \,\undersym{\lambda}{j}^{\alpha}\,
\undersym{\lambda}{i}^{\beta}.$
Consequently, $\frac{\partial g^{\mu\nu}}{\partial \, \undersym{\lambda}{j}_{\beta}} = - (g^{\mu\beta}
\,\undersym{\lambda}{j}^{\nu} + g^{\nu\beta}\,\undersym{\lambda}{j}^{\mu}).$

\item [(c)] $\frac{\partial ||\lambda||}{\partial \,\undersym{\lambda}{j}_{b}} = ||\lambda||\,\undersym{\lambda}{j}^{b},$ \ $||\lambda||: = det(\lambda_{b})$.
Consequently, \ $||\lambda||\!\,_{; a} = ||\lambda||C^{b}_{ba},$ \ $\delta_{\gamma}||\lambda|| = ||\lambda||\Gamma^{b}_{\gamma b}.$

\vspace{- 0.1 cm}\item [(d)] $\frac{\partial \,\undersym{\lambda}{i}^{a}}{\partial \,\undersym{\lambda}{j}_{b}} = - \,\undersym{\lambda}{j}^{a}\,
\undersym{\lambda}{i}^{b}.$ Consequently, $\frac{\partial g^{cd}}{\partial \, \undersym{\lambda}{j}_{b}} = - (g^{cb}
\,\undersym{\lambda}{j}^{d} + g^{db}\,\undersym{\lambda}{j}^{c}).$

\item [(e)] $\frac{\partial C^{a}_{dc}}{\partial\, \undersym{\lambda}{j}_{b}} =
- \,\undersym{\lambda}{j}^{a}C^{b}_{dc}$. Consequently, $\frac{\partial T^{a}_{dc}}{\partial\, \undersym{\lambda}{j}_{b}} =
\,\undersym{\lambda}{j}^{a}T^{b}_{cd}$, \ \ $\frac{\partial C_{d}}{\partial\, \undersym{\lambda}{j}_{b}} =
\,\undersym{\lambda}{j}^{a}T^{b}_{ad}.$

\item [(f)] $\frac{\partial C^{a}_{dc}}{\partial\, \undersym{\lambda}{j}_{b; e}} =
\,\undersym{\lambda}{j}^{a}\delta^{b}_{d}\delta^{e}_{c}$. Consequently, $\frac{\partial T^{a}_{dc}}{\partial\, \undersym{\lambda}{j}_{b, e}} =
\,\undersym{\lambda}{j}^{a}(\delta^{b}_{d}\delta^{e}_{c} - \delta^{e}_{d}\delta^{b}_{c})$, \ \
$\frac{\partial C_{d}}{\partial\, \undersym{\lambda}{j}_{b; e}} =
\,\undersym{\lambda}{j}^{a}(\delta^{b}_{d}\delta^{e}_{a} - \delta^{e}_{d}\delta^{b}_{a}).$

\end{description}

\end{lemma}

\begin{lemma}\label{lem} Assume that the nonlinear connection $N^{a}_{\mu}$ does not depend on the horizontal
counterparts $\,\undersym{\lambda}{i}^{\alpha}$ of the fundamental vector fields. Then we have:
\begin{description}

\item [(a)] $\frac{\partial \Gamma^{\alpha}_{\mu\nu}}{\partial\, \undersym{\lambda}{j}_{\beta}} =
- \,\undersym{\lambda}{j}^{\alpha}\Gamma^{\beta}_{\mu\nu}$.
Consequently, $\frac{\partial \Lambda^{\alpha}_{\mu\nu}}{\partial\, \undersym{\lambda}{j}_{\beta}} =
\,\undersym{\lambda}{j}^{\alpha}\Lambda^{\beta}_{\nu\mu}$, \ \ $\frac{\partial C_{\mu}}{\partial\, \undersym{\lambda}{j}_{\beta}} =
\,\undersym{\lambda}{j}^{\alpha}\Lambda^{\beta}_{\alpha\mu}.$

\item [(b)] $\frac{\partial \Gamma^{\alpha}_{\mu\nu}}{\partial\, \undersym{\lambda}{j}_{\beta, \gamma}} =
\,\undersym{\lambda}{j}^{\alpha}\delta^{\beta}_{\mu}\delta^{\gamma}_{\nu}$.
Consequently, $\frac{\partial \Lambda^{\alpha}_{\mu\nu}}{\partial\, \undersym{\lambda}{j}_{\beta, \gamma}} =
\,\undersym{\lambda}{j}^{\alpha}(\delta^{\beta}_{\mu}\delta^{\gamma}_{\nu} - \delta^{\gamma}_{\mu}\delta^{\beta}_{\nu})$,
\ \ $\frac{\partial C_{\mu}}{\partial\, \undersym{\lambda}{j}_{\beta, \gamma}} =
\,\undersym{\lambda}{j}^{\alpha}(\delta^{\beta}_{\mu}\delta^{\gamma}_{\alpha} - \delta^{\gamma}_{\mu}\delta^{\beta}_{\alpha}).$

\item [(c)] $\frac{\partial \Gamma^{\alpha}_{\mu\nu}}{\partial\, \undersym{\lambda}{j}_{\beta; a}} = -
\,\undersym{\lambda}{j}^{\alpha}\delta^{\beta}_{\mu}N^{a}_{\nu}$.
\\Consequently, $\frac{\partial \Lambda^{\alpha}_{\mu\nu}}{\partial\, \undersym{\lambda}{j}_{\beta; a}} =
\,\undersym{\lambda}{j}^{\alpha}(\delta^{\beta}_{\nu}N^{a}_{\mu} - \delta^{\beta}_{\mu}N^{a}_{\nu})$, \ \
$\frac{\partial C_{\mu}}{\partial\, \undersym{\lambda}{j}_{\beta; a}} =
\,\undersym{\lambda}{j}^{\alpha}(\delta^{\beta}_{\alpha}N^{a}_{\mu} - \delta^{\beta}_{\mu}N^{a}_{\alpha}).$

\end{description}
\end{lemma}

\begin{proof} We prove {\bf (c)} only. The rest is similar. By hypothesis, $N^{a}_{\mu}$ is independent of $\lambda_{\beta}$ and its derivatives. Consequently,
\begin{eqnarray*}\frac{\partial \Gamma^{\alpha}_{\mu\nu}}{\partial \,\undersym{\lambda}{j}_{\beta; c}}&=&
\frac{\partial }{\partial \,\undersym{\lambda}{j}_{\beta; c}}(\,\undersym{\lambda}{i}^{\alpha}\,\undersym{\lambda}{i}_{\mu, \nu} -
\,\undersym{\lambda}{i}^{\alpha}N^{a}_{\nu}\,\undersym{\lambda}{i}_{\mu; a}) = - \frac{\partial }{\partial \,\undersym{\lambda}{j}_{\beta; c}}
(\,\undersym{\lambda}{i}^{\alpha}N^{a}_{\nu}\,\undersym{\lambda}{i}_{\mu; a})\\&=&  - \,\undersym{\lambda}{i}^{\alpha}N^{a}_{\nu}
\bigg(\frac{\partial \,\undersym{\lambda}{i}_{\mu; a}}{\partial \,\undersym{\lambda}{j}_{\beta; c}}\bigg) = - \,\undersym{\lambda}{i}^{\alpha}
N^{a}_{\nu}(\delta_{ij}\delta^{c}_{a}\delta^{\beta}_{\mu}) =  - \,\undersym{\lambda}{j}^{\alpha}\delta^{\beta}_{\mu}N^{c}_{\nu}\\
\\[- 2 cm]\end{eqnarray*}
\end{proof}

\begin{corollary} \label{form}Let ${\bf T} = (\Lambda^{\alpha}_{\mu\, \nu}, \,R^{a}_{\mu\nu},
\,C^{\alpha}_{\mu c}, \,P^{a}_{\mu c}, \,T^{a}_{bc})$ be the torsion of the EAP-space. If
$P^{a}_{\mu c} = 0$, then $N^{a}_{\mu; a} = \delta_{\mu}(ln ||\lambda||)$ is a scalar 1-form.
\end{corollary}

\begin{proof} By hypothesis, we have $N^{a}_{\mu; a} = \Gamma^{a}_{a\mu}$.
On the other hand, by Lemma \ref{lemma} {\bf (c)},
$\delta_{\mu}||\lambda|| = ||\lambda||\Gamma^{a}_{a\mu}$.
Consequently, $\delta_{\mu}(ln ||\lambda||) = \frac{1}{||\lambda||}\delta_{\mu}||\lambda|| = N^{a}_{\mu; a}.$
\end{proof}
\end{mysect}


\begin{mysect}{Unified Field Equations}
We now generalize the GFT in the context of the EAP-geometry. We
derive the field equations using a variational technique, which
involves the variation of an appropriate chosen (horizontal and
vertical) Lagrangians with respect to the (horizontal and vertical)
fundamental vector fields.


\subsection{Horizontal unified field equations}

\hspace{0.6 cm} Let $(TM, \lambda)$ be an EAP-space. {\it We assume that the nonlinear connection $N^{a}_{\mu}$ does not depend on the horizontal
fundamental vector fields} so that Lemma \ref{lem} holds.

\bigskip

We now formulate a generalized version of the GFT \cite{aaa} in the
framework of EAP-geometry under the above mentioned condition. The
problem is to find an appropriate scalar Lagrangian that can work
effectively in the context of EAP-geometry. We make the following
choice. We take for the horizontal field equations a Lagrangian
similar in form (but not in content) to that used by Mikhail and
Wanas in their construction of the GFT. This is done for three
reasons. First, the form of the chosen Lagrangian is relatively
simple (depends on the first derivatives of the horizontal
counterparts of the fundamental vector fields). Second, this form of
the Lagrangian, in the conventional AP-context, has led to powerful
theoretical and experimental results. Last, in order to facilitate
comparison between the results obtained in our unified field theory
and the GFT.
\bigskip

In view of the above, we start with the following scalar Lagrangian:
Let
$${\cal H} = |\lambda| g^{\mu\nu} H_{\mu\nu},$$ where
\begin{equation}\label{rzx}H_{\mu\nu} := \Lambda^{\alpha}_{\epsilon\mu}\Lambda^{\epsilon}_{\alpha\nu} - C_{\mu}C_{\nu}.\end{equation}

The Euler-Lagrange equations \cite{GLS} for this Lagrangian are given by
\begin{equation}\label{ELEst}\frac{\delta {\cal H}}{\delta \lambda_{\beta}}: = \frac{\partial {\cal H}}{\partial \lambda_{\beta}} -
\frac{\partial}{\partial x^{\gamma}}\bigg(\frac{\partial {\cal H}}
{\partial \lambda_{\beta, \gamma}}\bigg) - \frac{\partial}{\partial y^{a}}\bigg(\frac{\partial {\cal H}}
{\partial \lambda_{\beta; a}}\bigg) = 0.\end{equation}

Setting \begin{equation}{\cal H} = {\cal K} - {\cal L};\nonumber\end{equation}
\begin{equation}{\cal K} := |\lambda|g^{\mu\nu}K_{\mu\nu}: = |\lambda|g^{\mu\nu}\Lambda^{\alpha}_{\epsilon\mu}
\Lambda^{\epsilon}_{\alpha\nu},\nonumber\end{equation}
\begin{equation}{\cal L} := |\lambda|g^{\mu\nu}L_{\mu\nu}:= |\lambda|g^{\mu\nu}C_{\mu}C_{\nu},\nonumber\end{equation}
(\ref{ELEst}) can be written in the form
\begin{equation}\frac{\delta {\cal H}}{\delta \lambda_{\beta}} = \frac{\delta {\cal K}}{\delta \lambda_{\beta}} -
\frac{\delta {\cal L}}{\delta \lambda_{\beta}}.\nonumber\end{equation}


We first consider the expression
$$\frac{\delta {\cal L}}{\delta \lambda_{\beta}} :=\frac{\partial {\cal L}}{\partial \lambda_{\beta}} -
\frac{\partial}{\partial x^{\gamma}}\bigg(\frac{\partial {\cal L}}
{\partial \lambda_{\beta, \gamma}}\bigg) - \frac{\partial}{\partial y^{a}}\bigg(\frac{\partial {\cal L}}
{\partial \lambda_{\beta; a}}\bigg).$$

Taking into account Lemmas \ref{lemma} and \ref{lem}, we have
\begin{equation}\frac{\partial {\cal L}}{\partial \lambda_{\beta}}= |\lambda|\{\lambda^{\beta}g^{\mu\nu}L_{\mu\nu} -
(g^{\mu\beta}\lambda^{\nu} + g^{\nu\beta}\lambda^{\mu})L_{\mu\nu} + 2g^{\mu\nu}
(\lambda^{\alpha}\Lambda^{\beta}_{\alpha\mu}C_{\nu})\}.\nonumber\end{equation}

\hspace{- 0.6 cm}Consequently, by the relation $\,\undersym{\lambda}{j}^{\alpha}\,\undersym{\lambda}{j}_{\sigma} = \delta^{\alpha}_{\sigma}$, we
deduce that
\begin{equation}\label{O}\frac{1}{|\lambda|}\bigg(\frac{\partial {\cal L}}{\partial \ \undersym{\lambda}{j}_{\beta}}\bigg) \ \undersym{\lambda}{j}_{\sigma}
= \delta^{\beta}_{\sigma}L - 2L^{\beta}_{\sigma} +
2\Lambda^{\beta}_{\sigma\mu}C^{\mu},\end{equation}
where $L:= g^{\mu\nu}L_{\mu\nu}$.

\bigskip

Moreover, by Lemma \ref{lemma}, we obtain
\begin{equation}\frac{\partial {\cal L}}{\partial \lambda_{\beta, \gamma}} = 2|\lambda|(\lambda^{\gamma}C^{\beta} - \lambda^{\beta}C^{\gamma}).
\nonumber\end{equation}
Hence, noting that, by Lemma \ref{lemma} {\bf (a)}, $|\lambda|_{,\gamma} = |\lambda|(\,\undersym{\lambda}{k}^{\mu}\,
\undersym{\lambda}{k}_{\mu, \gamma})$, we get

\begin{eqnarray*}\\[- 1.1 cm]\frac{\partial}{\partial x^{\gamma}}\bigg(\frac{\partial {\cal L}}{\partial \lambda_{\beta, \gamma}}\bigg)&=&
2|\lambda|\{(\, \undersym{\lambda}{k}^{\mu}\,
\undersym{\lambda}{k}_{\mu, \gamma})(\lambda^{\gamma}C^{\beta} - \lambda^{\beta}C^{\gamma}) +
(\lambda^{\gamma}C^{\beta})\!\,_{,\gamma} - (\lambda^{\beta}C^{\gamma})\!\,_{,\gamma}\},\\
\\[- 1.2 cm]\end{eqnarray*}
so that
\begin{equation}\begin{split}\label{SUB}\frac{1}{|\lambda|}\bigg\{\frac{\partial}
{\partial x^{\gamma}}\bigg(\frac{\partial {\cal L}}{\partial \,\undersym{\lambda}{j}_{\beta, \gamma}}\bigg)\bigg\}\, \undersym{\lambda}{j}_{\sigma}&=
2\{(\, \undersym{\lambda}{j}^{\mu}\, \undersym{\lambda}{j}_{\mu, \gamma})(\delta^{\gamma}_{\sigma}C^{\beta} - \delta^{\beta}_{\sigma}C^{\gamma})
+ [(\,\undersym{\lambda}{j}_{\sigma}\, \undersym{\lambda}{j}^{\gamma}\!\,_{,\gamma})C^{\beta}\\& \ \ \  + \ \delta^{\gamma}_{\sigma}C^{\beta}\!\,_{,\gamma}
- (\,\undersym{\lambda}{j}_{\sigma}\,\undersym{\lambda}{j}^{\beta}\!\,_{,\gamma})C^{\gamma} - \delta^{\beta}_{\sigma}C^{\gamma}\!\,_{,\gamma})]\}.
\end{split}\end{equation}

\vspace{- 0.1 cm}
\hspace{- 0.6 cm}Since $\delta_{\mu} = \partial_{\mu} - N^{a}_{\mu}\dot{\partial_a}$, \,$\Gamma^{\alpha}_{\mu\nu} = \undersym{\lambda}{j}^{\alpha}
\delta_{\nu}\,\undersym{\lambda}{j}_{\mu}$ and $C^{\alpha}_{\mu a} = \undersym{\lambda}{j}^{\alpha}
\dot{\partial_a}\,\undersym{\lambda}{j}_{\mu}$, it follows that
\begin{equation}\label{ab}\undersym{\lambda}{j}^{\mu}\, \undersym{\lambda}{j}_{\mu, \gamma} =
(\,\undersym{\lambda}{j}^{\mu}\delta_{\gamma}\,\undersym{\lambda}{j}_{\mu}) +
N^{a}_{\gamma}(\,\undersym{\lambda}{j}^{\mu}\dot{\partial_a}\, \undersym{\lambda}{j}_{\mu}) =
\Gamma^{\mu}_{\mu\gamma} + N^{a}_{\gamma}C^{\mu}_{\mu a},\end{equation}
\begin{equation}\label{b}\undersym{\lambda}{j}_{\sigma}\,\undersym{\lambda}{j}^{\gamma}\!\,_{,\gamma} = -
(\,\undersym{\lambda}{j}^{\gamma}\, \undersym{\lambda}{j}_{\sigma, \gamma}) = - (\Gamma^{\gamma}_{\sigma\gamma} + N^{a}_{\gamma}
C^{\gamma}_{\sigma a}),\end{equation}
\begin{equation}\label{c} C^{\beta}\!\,_{,\gamma} = \delta_{\gamma}C^{\beta} + N^{a}_{\gamma}\dot{\partial_a}C^{\beta}; \ \ \ \ \
C^{\gamma}\!\,_{,\gamma} = \delta_{\gamma}C^{\gamma} + N^{a}_{\gamma}\dot{\partial_a}C^{\gamma},\end{equation}
\begin{equation}\label{d}\undersym{\lambda}{j}_{\sigma}\,\undersym{\lambda}{j}^{\beta}\!\,_{,\gamma} = -
(\Gamma^{\beta}_{\sigma\gamma} + N^{a}_{\gamma}C^{\beta}_{\sigma a}).\end{equation}
\vspace{0.12 cm}
Substituting (\ref{ab}), (\ref{b}), (\ref{c}) and (\ref{d}) in (\ref{SUB}), we find that
\begin{equation}\begin{split}\label{AA}\frac{1}{|\lambda|}\bigg\{\frac{\partial}
{\partial x^{\gamma}}\bigg(\frac{\partial {\cal L}}{\partial \,\undersym{\lambda}{j}_{\beta, \gamma}}\bigg)\bigg\}\, \undersym{\lambda}{j}_{\sigma}&=
2(C^{\beta}\!\,_{|\sigma} - C_{\sigma}C^{\beta} -
\delta^{\beta}_{\sigma}(C^{\gamma}\!\,_{|\gamma} - C^{\gamma}C_{\gamma})  + C^{\epsilon}\Lambda^{\beta}_{\sigma\epsilon})\\
& \ \ \ \, + \ 2 (N^{a}_{\sigma}C^{\mu}_{\mu a}C^{\beta} + N^{a}_{\gamma}C^{\beta}_{\sigma a}C^{\gamma} -
N^{a}_{\gamma}C^{\gamma}_{\sigma a}C^{\beta} - \delta^{\beta}_{\sigma}N^{a}_{\gamma}C^{\mu}_{\mu a}C^{\gamma})\\
& \ \ \ \, + \ 2(N^{a}_{\sigma}\dot{\partial_a}C^{\beta} - \delta^{\beta}_{\sigma}N^{a}_{\gamma}\dot{\partial_a}C^{\gamma}).\end{split}\end{equation}
\vspace{- 0.1cm}On the other hand, by Lemma \ref{lem}, we have
\begin{equation}\frac{\partial {\cal L}}{\partial \lambda_{\beta; a}} = 2|\lambda|(\lambda^{\beta}N^{a}_{\nu}C^{\nu} - \lambda^{\alpha}
N^{a}_{\alpha}C^{\beta}). \nonumber\end{equation}
Consequently, noting that, by Lemma \ref{lemma} {\bf (a)}, $|\lambda|_{;a} = |\lambda|C^{\mu}_{\mu a}$, we obtain
\begin{eqnarray*}\frac{1}{|\lambda|}\bigg\{\frac{\partial}{\partial y^{a}}\bigg(\frac{\partial {\cal L}}{\partial \,\undersym{\lambda}{j}_{\beta; a}}\bigg)\bigg\}
\,\undersym{\lambda}{j}_{\sigma}&=& 2C^{\mu}_{\mu a}(\delta^{\beta}_{\sigma}N^{a}_{\nu}C^{\nu} -
N^{a}_{\sigma}C^{\beta})\\&& + \ 2\{(\,\undersym{\lambda}{j}_{\sigma}\,\undersym{\lambda}{j}^{\beta}\!\,_{; a})
N^{a}_{\nu}C^{\nu}
+ \delta^{\beta}_{\sigma}N^{a}_{\nu; a}C^{\nu} + \delta^{\beta}_{\sigma}N^{a}_{\nu}C^{\nu}\!\,_{; a}\}\\&& - \  2\{(\undersym{\lambda}{j}_{\sigma}\,
\undersym{\lambda}{j}^{\alpha}\!\,_{; a})N^{a}_{\alpha}C^{\beta}
+ \delta^{\alpha}_{\sigma}N^{a}_{\alpha; a}C^{\beta} + \delta^{\alpha}_{\sigma}N^{a}_{\alpha}C^{\beta}\!\,_{; a}\}.\\
\\[- 1 cm]\end{eqnarray*}
Since $C^{\alpha}_{\mu a} = \, \undersym{\lambda}{j}^{\alpha}\,\undersym{\lambda}{j}_{\mu; a}$
(or $\, \undersym{\lambda}{j}_{\mu}\,\undersym{\lambda}{j}^{\alpha}\!\,_{; a} = - \, C^{\alpha}_{\mu a}$), it follows
from the above formula that

\begin{equation}\begin{split}\\[- 0.8 cm]\label{t}\frac{1}{|\lambda|}\bigg\{\frac{\partial}{\partial y^{a}}\bigg(\frac{\partial {\cal L}}{\partial \,
\undersym{\lambda}{j}_{\beta; a}}\bigg)\bigg\}
\,\undersym{\lambda}{j}_{\sigma}&= 2(N^{a}_{\alpha}C^{\alpha}_{\sigma a}C^{\beta} - N^{a}_{\sigma}C^{\mu}_{\mu a}C^{\beta} -
N^{a}_{\gamma}C^{\beta}_{\sigma a}C^{\gamma})\\
& \ \ \ + \, 2(\delta^{\beta}_{\sigma}N^{a}_{\gamma}C^{\mu}_{\mu a}C^{\gamma} + \delta^{\beta}_{\sigma}N^{a}_{\nu}C^{\nu}\!\!\,_{; a} -
N^{a}_{\sigma}C^{\beta}\!\!\,_{; a})\\
& \ \ \ + \ 2C^{\nu}(\delta^{\beta}_{\sigma}N^{a}_{\nu; a} - \delta^{\beta}_{\nu}N^{a}_{\sigma; a}).\end{split}\end{equation}

By (\ref{O}), (\ref{AA}) and (\ref{t}), canceling equal terms, we
conclude that
\begin{equation}\begin{split}\label{FL}\\[- 0.4 cm]\frac{1}{|\lambda|}\bigg(\frac{\delta {\cal L}}{\delta \,\undersym{\lambda}{j}_{\beta}}\bigg)\,
\undersym{\lambda}{j}_{\sigma} &=
\delta^{\beta}_{\sigma}L - 2L^{\beta}_{\sigma} + 2\Lambda^{\beta}_{\sigma\mu}C^{\mu} -
2(C^{\beta}\!\,_{|\sigma} + C^{\epsilon}\Lambda^{\beta}_{\sigma\epsilon}) + 2C_{\sigma}C^{\beta}\\& \ \ \ \  + 2\delta^{\beta}_{\sigma}
(C^{\gamma}\!\,_{{|}\gamma} - C^{\gamma}C_{\gamma}) - 2C^{\nu}(\delta^{\beta}_{\sigma}N^{a}_{\nu; a} - \delta^{\beta}_{\nu}
N^{a}_{\sigma; a}).\end{split}\end{equation}

We next consider the expression
$$\frac{\delta {\cal K}}{\delta \lambda_{\beta}} : = \frac{\partial {\cal K}}{\partial \lambda_{\beta}} -
\frac{\partial}{\partial x^{\gamma}}\bigg(\frac{\partial {\cal K}}
{\partial \lambda_{\beta, \gamma}}\bigg) - \frac{\partial}{\partial y^{a}}\bigg(\frac{\partial {\cal K}}
{\partial \lambda_{\beta; a}}\bigg).$$

Again, by Lemmas \ref{lemma} and \ref{lem}, we find that
\begin{eqnarray*}\frac{\partial {\cal K}}{\partial \lambda_{\beta}} &=& |\lambda|\big\{\lambda^{\beta}g^{\mu\nu}K_{\mu\nu} -
(g^{\mu\beta}\lambda^{\nu} + g^{\nu\beta}\lambda^{\mu})K_{\mu\nu} + g^{\mu\nu}
(\lambda^{\alpha}\Lambda^{\beta}_{\mu\epsilon}\Lambda^{\epsilon}_{\alpha\nu} +
\lambda^{\epsilon}\Lambda^{\alpha}_{\epsilon\mu}\Lambda^{\beta}_{\nu\alpha})\big\}.\\
\end{eqnarray*}
\\[-1 cm]Consequently,
\begin{equation}\label{OT}\frac{1}{|\lambda|}\bigg(\frac{\partial {\cal K}}{\partial \ \undersym{\lambda}{j}_{\beta}}\bigg) \ \undersym{\lambda}{j}_{\sigma}
= \delta^{\beta}_{\sigma}K - 2K^{\beta}_{\sigma} +
2g^{\mu\nu}\Lambda^{\beta}_{\mu\epsilon}\Lambda^{\epsilon}_{\sigma\nu},\end{equation}
where $K := g^{\mu\nu}K_{\mu\nu}.$

\bigskip
Moreover, by Lemma \ref{lem}, we have
\begin{equation}\frac{\partial {\cal K}}{\partial \lambda_{\beta, \gamma}} = 2|\lambda|(\lambda^{\epsilon}g^{\gamma\alpha}\Lambda^{\beta}_{\epsilon\alpha} -
\lambda^{\epsilon}g^{\alpha\beta}\Lambda^{\gamma}_{\epsilon\alpha}).\nonumber\end{equation}
Hence, noting that $|\lambda|_{,\gamma} = (|\lambda|\,\undersym{\lambda}{k}^{\mu})\, \undersym{\lambda}{k}_{\mu, \gamma}$, we get
\begin{equation}\begin{split}\label{SUBT}\frac{1}{|\lambda|}\bigg\{\frac{\partial}
{\partial x^{\gamma}}\bigg(\frac{\partial {\cal K}}{\partial \,\undersym{\lambda}{j}_{\beta, \gamma}}\bigg)\bigg\}\, \undersym{\lambda}{j}_{\sigma}&= 2\{
(\, \undersym{\lambda}{k}^{\mu}\, \undersym{\lambda}{k}_{\mu, \gamma})(\delta^{\epsilon}_{\sigma}\,g^{\gamma\alpha}\Lambda^{\beta}_{\epsilon\alpha} -
\delta^{\epsilon}_{\sigma}\,g^{\alpha\beta}\Lambda^{\gamma}_{\epsilon\alpha})\}\\
& \ \ \, \, + \ 2\{(\,\undersym{\lambda}{j}_{\sigma}\,\undersym{\lambda}{j}^{\epsilon}\!\,_{,\gamma})g^{\gamma\alpha}\Lambda^{\beta}_{\epsilon\alpha}
+ \delta^{\epsilon}_{\sigma}\,g^{\gamma\alpha}\!\,_{,\gamma}
\Lambda^{\beta}_{\epsilon\alpha} + \delta^{\epsilon}_{\sigma}\,g^{\gamma\alpha}\Lambda^{\beta}_{\epsilon\alpha, \gamma}\}\\
& \ \ \, \, - \
2\{(\,\undersym{\lambda}{j}_{\sigma}\,\undersym{\lambda}{j}^{\epsilon}\!\,_{,\gamma})g^{\alpha\beta}\Lambda^{\gamma}_{\epsilon\alpha}
+ \delta^{\epsilon}_{\sigma}g^{\alpha\beta}\!\,_{,\gamma}
\,\Lambda^{\gamma}_{\epsilon\alpha} +
\delta^{\epsilon}_{\sigma}\,g^{\alpha\beta}\Lambda^{\gamma}_{\epsilon\alpha,
\gamma}\}.\end{split}\end{equation} As easily checked,
\begin{equation}\label{abT}\undersym{\lambda}{j}^{\mu}\, \undersym{\lambda}{j}_{\mu, \gamma} = \Gamma^{\mu}_{\mu\gamma} +
N^{a}_{\gamma}C^{\mu}_{\mu a}; \ \ \ \ \ \ \ \ \ \undersym{\lambda}{j}_{\sigma}\,\undersym{\lambda}{j}^{\epsilon}\,_{,\gamma} = -
(\Gamma^{\epsilon}_{\sigma\gamma} + N^{a}_{\gamma}
C^{\epsilon}_{\sigma a}),\end{equation}
\begin{equation}\label{cT} g^{\beta\alpha}\,_{,\gamma} = \delta_{\gamma}g^{\beta\alpha} + N^{a}_{\gamma}\dot{\partial_a}g^{\beta\alpha}; \ \ \ \ \
g^{\gamma\alpha}\,_{,\gamma} = \delta_{\gamma}g^{\gamma\alpha} + N^{a}_{\gamma}\dot{\partial_a}g^{\gamma\alpha},\end{equation}
\begin{equation}\label{dT} \Lambda^{\beta}_{\epsilon\alpha, \gamma} = \delta_{\gamma}\Lambda^{\beta}_{\epsilon\alpha} +
N^{a}_{\gamma}\dot{\partial_a}\Lambda^{\beta}_{\epsilon\alpha}; \ \ \ \ \
\Lambda^{\gamma}_{\epsilon\alpha, \gamma} = \delta_{\gamma}\Lambda^{\gamma}_{\epsilon\alpha} +
N^{a}_{\gamma}\dot{\partial_a}\Lambda^{\gamma}_{\epsilon\alpha}.\end{equation}

\vspace{0.2 cm}
\hspace{- 0.6 cm}Substituting (\ref{abT}), (\ref{cT}) and (\ref{dT}) in (\ref{SUBT}), we obtain
\begin{equation}\begin{split}\label{SUBTT}\\[ - 0.3 cm]\frac{1}{|\lambda|}\bigg\{\frac{\partial}
{\partial x^{\gamma}}\bigg(\frac{\partial {\cal K}}{\partial \,\undersym{\lambda}{j}_{\beta, \gamma}}\bigg)\bigg\}\, \undersym{\lambda}{j}_{\sigma} &=
2\big\{\Gamma^{\mu}_{\mu\gamma}(g^{\gamma\alpha}\Lambda^{\beta}_{\sigma\alpha} - g^{\beta\alpha}\Lambda^{\gamma}_{\sigma\alpha}) -
\Gamma^{\epsilon}_{\sigma\gamma}(g^{\gamma\alpha}\Lambda^{\beta}_{\epsilon\alpha} - g^{\beta\alpha}\Lambda^{\gamma}_{\epsilon\alpha})\big\}
\\& \ \ \ \, + \ 2\big\{(g^{\gamma\alpha}\delta_{\gamma}\Lambda^{\beta}_{\sigma\alpha} -
g^{\beta\alpha}\delta_{\gamma}\Lambda^{\gamma}_{\sigma\alpha}) - (\delta_{\gamma}
g^{\beta\alpha}\Lambda^{\gamma}_{\sigma\alpha} - \delta_{\gamma}g^{\gamma\alpha}\Lambda^{\beta}_{\sigma\alpha})\big\}\\
& \ \ \ \, + \, 2\big\{N^{a}_{\gamma}C^{\mu}_{\mu a}(g^{\gamma\alpha}\Lambda^{\beta}_{\sigma\alpha} - g^{\alpha\beta}\Lambda^{\gamma}_{\sigma\alpha})
- N^{a}_{\gamma}C^{\epsilon}_{\sigma a}(g^{\gamma\alpha}\Lambda^{\beta}_{\epsilon\alpha} - g^{\beta\alpha}
\Lambda^{\gamma}_{\epsilon\alpha})\big\}\\
& \ \ \ \, + \, 2N^{a}_{\gamma}\big\{(g^{\gamma\alpha}\dot{\partial_a}\Lambda^{\beta}_{\sigma\alpha} - g^{\beta\alpha}
\dot{\partial_a}\Lambda^{\gamma}_{\sigma\alpha}) - (\Lambda^{\gamma}_{\sigma\alpha}\dot{\partial_a}g^{\beta\alpha}
- \Lambda^{\beta}_{\sigma\alpha}\dot{\partial_a}g^{\gamma\alpha})\big\}.\end{split}\end{equation}

\vspace{- 0.1 cm}
On the other hand, Lemma \ref{lem} gives
\begin{eqnarray*} \frac{\partial {\cal K}}{\partial \lambda_{\beta; a}} &=& 2|\lambda|\lambda^{\epsilon}(g^{\beta\mu}N^{a}_{\alpha}\Lambda^{\alpha}_{\epsilon\mu} -
g^{\mu\nu}N^{a}_{\nu}\Lambda^{\beta}_{\epsilon\mu}).\\
\\[- 1 cm]\end{eqnarray*}
Consequently, taking into account that $|\lambda|_{;a} = |\lambda|C^{\mu}_{\mu a}$, we get
\begin{eqnarray*}\frac{1}{|\lambda|}\bigg\{\frac{\partial}{\partial y^{a}}\bigg(\frac{\partial {\cal K}}{\partial \,\undersym{\lambda}{j}_{\beta; a}}\bigg)\bigg\}
\,\undersym{\lambda}{j}_{\sigma}&=&2C^{\delta}_{\delta a}
(\delta^{\epsilon}_{\sigma}g^{\beta\mu}N^{a}_{\alpha}\Lambda^{\alpha}_{\epsilon\mu} -
\delta^{\epsilon}_{\sigma} g^{\mu\nu}N^{a}_{\nu}\Lambda^{\beta}_{\epsilon\mu})\,\undersym{\lambda}{j}_{\sigma}\\
&& +  \ 2(\,\undersym{\lambda}{j}_{\sigma}\,\undersym{\lambda}{j}^{\epsilon}\!\,_{; a})\{g^{\beta\mu}N^{a}_{\alpha}
\Lambda^{\alpha}_{\epsilon\mu} - g^{\mu\nu}N^{a}_{\nu}\Lambda^{\beta}_{\epsilon\mu}\} +
\delta^{\epsilon}_{\sigma}g^{\beta\mu}\!\,_{;a}N^{a}_{\alpha}\Lambda^{\alpha}_{\epsilon\mu}\\
&& + \ \delta^{\epsilon}_{\sigma}g^{\mu\nu}\!\,_{;a}N^{a}_{\nu}\Lambda^{\beta}_{\epsilon\mu} +
\delta^{\epsilon}_{\sigma}g^{\beta\mu}N^{a}_{\alpha; a}\Lambda^{\alpha}_{\epsilon\mu} +
\delta^{\epsilon}_{\sigma}g^{\mu\nu}N^{a}_{\nu; a}\Lambda^{\beta}_{\epsilon\mu}\\
&& + \ \delta^{\epsilon}_{\sigma}g^{\beta\mu}N^{a}_{\alpha}
\Lambda^{\alpha}_{\epsilon\mu; a}  + \ \delta^{\epsilon}_{\sigma}g^{\mu\nu}N^{a}_{\nu}\Lambda^{\beta}_{\epsilon\mu; a}.\\
\\[- 1.2 cm]\end{eqnarray*}

\hspace{- 0.6 cm}Moreover, since $C^{\alpha}_{\mu a} = \, \undersym{\lambda}{j}^{\alpha}\,\undersym{\lambda}{j}_{\mu; a}$, it follows that
\begin{equation}\begin{split}\label{SUBTTT}\\[ - 0.5 cm]\frac{1}{|\lambda|}\bigg\{\frac{\partial}
{\partial y^{a}}\bigg(\frac{\partial {\cal K}}{\partial \,\undersym{\lambda}{j}_{\beta; a}}\bigg)\bigg\}\, \undersym{\lambda}{j}_{\sigma} &=
2\{(g^{\beta\mu}N^{a}_{\alpha; a}\Lambda^{\alpha}_{\sigma\mu}-
g^{\mu\nu}N^{a}_{\nu; a}\Lambda^{\beta}_{\sigma\mu})\\
&\ \ \ \,  - \ N^{a}_{\gamma}C^{\mu}_{\mu a}(g^{\gamma\alpha}\Lambda^{\beta}_{\sigma\alpha} -
g^{\alpha\beta}\Lambda^{\gamma}_{\sigma\alpha})  + N^{a}_{\gamma}C^{\epsilon}_{\sigma a}(g^{\gamma\alpha}\Lambda^{\beta}_{\epsilon\alpha}
- g^{\beta\alpha}\Lambda^{\gamma}_{\epsilon\alpha})\\
&\ \ \ \,  - \ N^{a}_{\gamma}(g^{\gamma\alpha}\dot{\partial_a}\Lambda^{\beta}_{\sigma\alpha} - g^{\beta\alpha}
\dot{\partial_a}\Lambda^{\gamma}_{\sigma\alpha}) + N^{a}_{\gamma}(\Lambda^{\gamma}_{\sigma\alpha}\dot{\partial_a}g^{\beta\alpha}
 - \Lambda^{\beta}_{\sigma\alpha}\dot{\partial_a}g^{\gamma\alpha})\}.\end{split}\end{equation}

\vspace{- 0.3 cm}

By (\ref{OT}), (\ref{SUBTT}) and (\ref{SUBTTT}), after some reductions, we finally arrive at the relation

\begin{equation}\begin{split}\label{SS}\\[- 1 cm]
\frac{1}{|\lambda|}\bigg(\frac{\delta {\cal K}}{\delta \,\undersym{\lambda}{j}_{\beta}}\bigg)\,\undersym{\lambda}{j}_{\sigma} &=
\delta^{\beta}_{\sigma}K - 2K^{\beta}_{\sigma} +
2g^{\mu\nu}\Lambda^{\beta}_{\mu\epsilon}\Lambda^{\epsilon}_{\sigma\nu} -
2\Gamma^{\mu}_{\mu\gamma}(g^{\gamma\alpha}\Lambda^{\beta}_{\sigma\alpha} - g^{\beta\alpha}\Lambda^{\gamma}_{\sigma\alpha})\\
& \ \ \ \ + \ 2\Gamma^{\epsilon}_{\sigma\gamma}(g^{\gamma\alpha}\Lambda^{\beta}_{\epsilon\alpha} - g^{\beta\alpha}\Lambda^{\gamma}_{\epsilon\alpha})
- 2(g^{\gamma\alpha}\delta_{\gamma}\Lambda^{\beta}_{\sigma\alpha} -
g^{\beta\alpha}\delta_{\gamma}\Lambda^{\gamma}_{\sigma\alpha})\\
& \ \ \  \  + \ 2(\delta_{\gamma}
g^{\beta\alpha}\Lambda^{\gamma}_{\sigma\alpha} -   \delta_{\gamma}g^{\gamma\alpha}\Lambda^{\beta}_{\sigma\alpha}) -
2(g^{\beta\mu}N^{a}_{\alpha; a}\Lambda^{\alpha}_{\sigma\mu} - g^{\mu\nu}N^{a}_{\nu; a}\Lambda^{\beta}_{\sigma\mu}).\end{split}\end{equation}

In what follows, we shall derive some formulae to simplify (\ref{SS}). Recall that the cannonical $d$-connection is a metric connection.

\bigskip

We have $g^{\alpha\beta}\,_{|\gamma} = 0$, so that $\Lambda^{\gamma}_{\sigma\alpha}g^{\alpha\beta}\,_{|\gamma}  = 0$. Consequently,
\begin{equation}\label{x}\Lambda^{\gamma}_{\sigma\alpha}\delta_{\gamma}g^{\alpha\beta} = -
\Lambda^{\gamma}_{\sigma\alpha}g^{\alpha\epsilon}\Gamma^{\beta}_{\epsilon\gamma} - \Lambda^{\gamma}_{\sigma\epsilon}g^{\beta\alpha}
\Gamma^{\epsilon}_{\alpha\gamma}.\end{equation}
Next,
$$g^{\alpha\beta}\Lambda^{\gamma}\,_{\sigma\alpha|\gamma} =
g^{\alpha\beta}(\delta_{\gamma}\Lambda^{\gamma}_{\sigma\alpha} +
\Lambda^{\epsilon}_{\sigma\alpha}\Gamma^{\gamma}_{\epsilon\gamma} -
\Lambda^{\gamma}_{\epsilon\alpha}\Gamma^{\epsilon}_{\sigma\gamma} - \Lambda^{\gamma}_{\sigma\epsilon}
\Gamma^{\epsilon}_{\alpha\gamma}).\\[0.3 cm]$$
Hence, noting that $\Gamma^{\gamma}_{\epsilon\gamma} =
\Lambda^{\gamma}_{\epsilon\gamma} + \Gamma^{\gamma}_{\gamma\epsilon} = C_{\epsilon} + \Gamma^{\gamma}_{\gamma\epsilon}$, we
get
\begin{equation}\label{xt} g^{\alpha\beta}(\Lambda^{\gamma}\,_{\sigma\alpha|\gamma} - C_{\epsilon}\Lambda^{\epsilon}_{\sigma\alpha}) =
g^{\alpha\beta}(\delta_{\gamma}\Lambda^{\gamma}_{\sigma\alpha} +
\Lambda^{\epsilon}_{\sigma\alpha}\Gamma^{\gamma}_{\gamma\epsilon} -
\Lambda^{\gamma}_{\epsilon\alpha}\Gamma^{\epsilon}_{\sigma\gamma} - \Lambda^{\gamma}_{\sigma\epsilon}
\Gamma^{\epsilon}_{\alpha\gamma}).\end{equation}

\hspace{- 0.6 cm}Moreover,
\begin{equation}\label{y} g^{\gamma\alpha}\Lambda^{\beta}\,_{\sigma\alpha|\gamma} = g^{\gamma\alpha}\delta_{\gamma}
\Lambda^{\beta}_{\sigma\alpha} + g^{\alpha\epsilon}\Lambda^{\gamma}_{\sigma\alpha}\Gamma^{\beta}_{\gamma\epsilon} -
g^{\gamma\alpha}\Lambda^{\beta}_{\sigma\epsilon}\Gamma^{\epsilon}_{\alpha\gamma} - g^{\gamma\alpha}\Lambda^{\beta}_{\epsilon\alpha}
\Gamma^{\epsilon}_{\sigma\gamma}.\end{equation}
Adding the first term on the right hand side of (\ref{x}) and the second term on the right hand side of (\ref{y}), we obtain
\begin{equation}\label{xyz}g^{\alpha\epsilon}\Lambda^{\gamma}_{\sigma\alpha}(\Gamma^{\beta}_{\gamma\epsilon} - \Gamma^{\beta}_{\epsilon\gamma})
=  g^{\alpha\epsilon}\Lambda^{\gamma}_{\sigma\alpha}\Lambda^{\beta}_{\gamma\epsilon}=  g^{\mu\nu}\Lambda^{\beta}_{\epsilon\mu}
\Lambda^{\epsilon}_{\sigma\nu}.\end{equation}

\hspace{- 0.6 cm}As easily checked,
\begin{equation}\delta_{\gamma}g^{\gamma\alpha} + g^{\gamma\epsilon}\Gamma^{\alpha}_{\gamma\epsilon} + g^{\alpha\epsilon}\Gamma^{\gamma}_{
\epsilon\gamma} = 0.\nonumber\end{equation}

\hspace{- 0.6 cm}Consequently,
\begin{equation}\begin{split}\label{TAR} 0 &= \Lambda^{\beta}_{\sigma\alpha}(\delta_{\gamma}g^{\gamma\alpha} +
g^{\gamma\epsilon}\Gamma^{\alpha}_{\gamma\epsilon} + g^{\alpha\epsilon}\Gamma^{\gamma}_{\epsilon\gamma})\\
& =  \Lambda^{\beta}_{\sigma\alpha}\delta_{\gamma}g^{\gamma\alpha} +
\Lambda^{\beta}_{\sigma\alpha}g^{\gamma\epsilon}\Gamma^{\alpha}_{\gamma\epsilon} +
\Lambda^{\beta}_{\sigma\alpha}g^{\alpha\epsilon}\Gamma^{\gamma}_{\gamma\epsilon} +
\Lambda^{\beta}_{\sigma\alpha}C^{\alpha}.\end{split}\end{equation}

\hspace{- 0.6 cm}Finally, it is clear that the third term on the right hand side of (\ref{y}) cancels with the second term on the right hand side of (\ref{TAR}), that is
\begin{equation}g^{\gamma\epsilon}\Lambda^{\beta}_{\sigma\alpha}\Gamma^{\alpha}_{\gamma\epsilon} -
g^{\gamma\alpha}\Lambda^{\beta}_{\sigma\epsilon}\Gamma^{\epsilon}_{\alpha\gamma} = 0\nonumber\end{equation}

In view of (\ref{xt}), (\ref{xyz}), and (\ref{TAR}), equation (\ref{SS}) may be written in the form
\begin{equation}\begin{split}\label{FES}\\[- 0.3 cm]\frac{1}{|\lambda|}\bigg(\frac{\delta {\cal K}}
{\delta \,\undersym{\lambda}{j}_{\beta}}\bigg)\,\undersym{\lambda}{j}_{\sigma} &=
\delta^{\beta}_{\sigma}K - 2K^{\beta}_{\sigma} + 2g^{\alpha\beta}(\Lambda^{\gamma}\,_{\sigma\alpha|\gamma} -
C_{\epsilon}\Lambda^{\epsilon}_{\sigma\alpha}) - 2g^{\gamma\alpha}\Lambda^{\beta}_{\sigma\alpha|\gamma}\\
& \ \ \ \ + 2\Lambda^{\beta}_{\sigma\alpha}C^{\alpha} - 2(g^{\beta\mu}N^{a}_{\alpha; a}\Lambda^{\alpha}_{\sigma\mu} -
g^{\mu\nu}N^{a}_{\nu; a}\Lambda^{\beta}_{\sigma\mu}).\end{split}\end{equation}

We finally consider the Euler-Lagrange equations (\ref{ELEst}). Setting
\begin{equation}\label{xcx}E^{\beta}_{\sigma}: = \frac{1}{|\lambda|}\bigg(\frac{\delta {\cal H}}{\delta \,\undersym{\lambda}{j}_{\beta}}\bigg)\,
\undersym{\lambda}{j}_{\sigma},\end{equation}
then, according to (\ref{FL}), (\ref{FES}) and (\ref{xcx}), we conclude that
\begin{equation}\begin{split}\label{FFFL}\\[- 0.45 cm]E^{\beta}_{\sigma} &= \delta^{\beta}_{\sigma}H - 2H^{\beta}_{\sigma} +
2(C^{\beta}\!\,_{|\sigma} + C^{\epsilon}\Lambda^{\beta}_{\sigma\epsilon}) - 2C_{\sigma}C^{\beta} - 2\delta^{\beta}_{\sigma}
(C^{\gamma}\!\,_{{|}\gamma} - C^{\gamma}C_{\gamma})  - 2g^{\gamma\alpha}\Lambda^{\beta}_{\sigma\alpha|\gamma}
\\& \ \ \ +
2g^{\alpha\beta}(\Lambda^{\gamma}\,_{\sigma\alpha|\gamma} - C_{\epsilon}\Lambda^{\epsilon}_{\sigma\alpha})
- 2(g^{\beta\mu}N^{a}_{\alpha; a}\Lambda^{\alpha}_{\sigma\mu}
- g^{\mu\nu}N^{a}_{\nu; a}\Lambda^{\beta}_{\sigma\mu})\\& \ \ \ + 2C^{\nu}(\delta^{\beta}_{\sigma}N^{a}_{\nu; a} -
\delta^{\beta}_{\nu}N^{a}_{\sigma; a}).\end{split}\end{equation}
On the other hand, by (\ref{h}), we have
\begin{equation}\label{PPP}C_{\sigma|\alpha} - C_{\alpha|\sigma} = \Lambda^{\gamma}_{\sigma\alpha|\gamma} -
C_{\epsilon}\Lambda^{\epsilon}_{\sigma\alpha} + \mathfrak{S}_{\sigma, \alpha, \epsilon}C^{\epsilon}_{\alpha a}R^{a}_{\epsilon\sigma}.\end{equation}

Hence, by (\ref{PPP}), we obtain
\begin{equation}\begin{split}\label{Aida} (C^{\beta}\,_{{|}\sigma} + C^{\epsilon}\Lambda^{\beta}_{\sigma\epsilon}) + g^{\alpha\beta}
(\Lambda^{\gamma}\,_{\sigma\alpha|\gamma} - C_{\epsilon}\Lambda^{\epsilon}_{\sigma\alpha})
&= g^{\alpha\beta}C_{\sigma|\alpha} - C^{\epsilon}\Lambda^{\beta}_{\epsilon\sigma}\\& \ \ \ +
g^{\alpha\beta}\{\mathfrak{S}_{\sigma, \alpha, \epsilon}C^{\epsilon}_{\alpha a}R^{a}_{\sigma\epsilon}\}
\end{split}\end{equation}
In view of (\ref{FFFL}) and (\ref{Aida}), the Euler-Lagrange equations  take the form
\begin{equation}\begin{split}\label{EEEE}\\[- 0.5 cm]0&=E^{\beta}_{\sigma} = \delta^{\beta}_{\sigma}H - 2H^{\beta}_{\sigma}
- 2C_{\sigma}C^{\beta} - 2\delta^{\beta}_{\sigma}C^{\epsilon}\!\,_{|\epsilon} +
2\delta^{\beta}_{\sigma}C^{\epsilon}C_{\epsilon} -2C^{\epsilon}\Lambda^{\beta}_{\epsilon\sigma}\\& \ \ \ \ \ \ \ \ \ \ \
\ + \, 2g^{\alpha\beta}C_{\sigma|\alpha} - 2g^{\gamma\alpha}\Lambda^{\beta}_{\sigma\alpha|\gamma}
 - 2N^{a}_{\alpha; a}(\Lambda^{\alpha}\,_{\sigma}\,^{\beta}-  \Lambda^{\beta}\,_{\sigma}\,^{\alpha})\\&
\ \ \ \ \ \ \ \ \ \ \ \ + \, 2C^{\nu}(\delta^{\beta}_{\sigma}N^{a}_{\nu; a} - \delta^{\beta}_{\nu}N^{a}_{\sigma; a}) + 2g^{\alpha\beta}\{\mathfrak{S}_{\sigma, \alpha, \epsilon}C^{\epsilon}_{\alpha a}
R^{a}_{\sigma\epsilon}\}\end{split}\end{equation}
Lowering the index $\beta$ in (\ref{EEEE}) and renaming the indices, we get
\begin{equation}\begin{split}\label{EEEXz}\\[- 0.7 cm]0& = E_{\mu\nu} := g_{\mu\nu}H - 2H_{\mu\nu} -
2C_{\mu}C_{\nu} - 2g_{\mu\nu}(C^{\epsilon}\!\,_{|\epsilon} - C^{\epsilon}C_{\epsilon})- 2C^{\epsilon}\Lambda_{\mu\epsilon\nu}
+ 2C_{\nu|\mu}\\& \ \ \ \ \ \ \ \ \ \ \ \ \ \ - \ 2g^{\epsilon\alpha}\Lambda_{\mu\nu\alpha|\epsilon}
 - 2N^{a}_{\epsilon; a}(\Lambda^{\epsilon}\,_{\nu\mu} -  \Lambda_{\mu\nu}\,^{\epsilon})
+ 2g_{\mu\nu}C^{\epsilon}N^{a}_{\epsilon; a} - 2C_{\mu}N^{a}_{\nu; a}\\& \ \ \ \ \ \ \ \ \ \ \ \ \ \ + 2 \ \mathfrak{S}_{\mu, \nu, \epsilon}C^{\epsilon}_{\mu a}
R^{a}_{\nu\epsilon}.\end{split}\end{equation}
This is the {\bf horizontal unified field equations} in the context of the EAP-geometry.

\bigskip

Although the chosen Lagrangian (\ref{rzx}) has a similar form to that of the GFT, the horizontal field equations derived in the
framework of EAP-geometry,
as should be expected, contain extra
terms which do not exist in the context of the GFT. This is because the Euler-Lagrange equations contain an additional term due to the
dependence of the horizontal fundamental vector fields on the directional argument $y$ (\ref{ELEst}). Moreover, these extra terms are
expressed explicitly in terms of the nonlinear connection of the EAP-space (\ref{EEEXz}).


\subsection{Vertical unified field equations}

\hspace{0.6 cm} We now deduce the vertical field equations. {\it No conditions are imposed on the nonlinear connection}.

\bigskip

As easily checked, Lemma \ref{ct} remains valid if all horizontal geometric objects are replaced by their vertical corespondents.

\bigskip
We consider here a scalar Lagrangian formed of vertical entities. Let
$${\cal V} := {||\lambda||} g^{ab} V_{ab},$$ where
\begin{equation}V_{ab} := T^{d}_{ea}T^{e}_{db} - C_{a}C_{b}.\nonumber\end{equation}

Since, by (\ref{cnt}) and (\ref{torsion}), the vertical torsion tensor $T^{a}_{bc}$ is expressed in terms of both $\lambda^{a}$ and $\lambda_{b; c}$, it follows that
${\cal V}$ does not depend on $\lambda_{b, \mu}$, that is
\begin{equation}\frac{\partial {\cal V}}{\partial \lambda_{b, \mu}} = 0.\nonumber\end{equation} Consequently,
the Euler-Lagrange equations in this case reduce to
\begin{equation}\frac{\partial {\cal V}}{\partial \lambda_{b}} - \frac{\partial}{\partial y^{e}}\bigg(\frac{\partial {\cal V}}
{\partial \lambda_{b; \, e}}\bigg) = 0.\nonumber\end{equation}
Following the same procedure of the proof of equation
(\ref{EEEXz}), taking into account Lemma \ref{lemma}, (with each geometric object being replaced by its vertical analogue), we get

\begin{equation}\begin{split}\label{xax} 0 &= E_{ab}: = g_{ab} V - 2V_{ab} - 2g_{ab} (C^{e}\!\,_{||e} - C^{e}C_{e})
- 2C_{a}C_{b} - 2C^{e} T_{aeb}\\& \ \ \ \ \ \ \ \ \ \ \ \ \ + 2C_{b||a} -
2g^{de}T_{abe||d}.\end{split}\end{equation}

\hspace{- 0.65 cm}This is the {\bf vertical unified field equations} in the context of the EAP-geometry.

\bigskip

It should be noted that equation (\ref{xax}) is formally similar to the field equations of the GFT with each geometric object in the context of the GFT replaced
by its vertical counterpart.
\end{mysect}


\begin{mysect}{Physical Consequences}

In this section, we investigate some physical consequences of the obtained field equations. To do this, we split the obtained field equations
into its symmetric and
skew-symmetric parts. We show that the symmetric
(skew-symmetric) part of the field equations give rise to a generalized form of Einstein's equations (Maxwell's equations).
Moreover, all physical objects considered are purely geometric.


\vspace{- 0.2 cm}\subsection{Splitting the horizontal field equations}

\hspace{0.6 cm} We first focus our attention on the symmetric and skew symmetric parts of the horizontal field equations (\ref{EEEXz}).

\bigskip

Considering the {\bf symmetric} part of (\ref{EEEXz}), noting that $C = C^{\epsilon}\!\,_{|\epsilon} - C^{\epsilon}C_{\epsilon}$,
we have
\begin{equation}\begin{split}\label{EEEXy}\\[- 0.7 cm]0& = E_{(\mu\nu)} = g_{\mu\nu}H - 2H_{\mu\nu} -  2C_{\mu}C_{\nu} - 2g_{\mu\nu}C -
C^{\epsilon}(\Lambda_{\mu\epsilon\nu} + \Lambda_{\nu\epsilon\mu}) + (C_{\nu|\mu} + C_{\mu|\nu})
\\& \ \ \ \ \ \ \ \ \ \ \ \ \ \, - \,g^{\epsilon\alpha}(\Lambda_{\mu\nu\alpha|\epsilon} + \Lambda_{\nu\mu\alpha|\epsilon})
 + N^{a}_{\epsilon; a}(\Lambda_{\mu\nu}\!\,^{\epsilon} +  \Lambda_{\nu\mu}\!\,^{\epsilon})  + 2g_{\mu\nu}C^{\epsilon}N^{a}_{\epsilon; a}\\&
\ \ \ \ \ \ \ \ \ \ \ \ \ \  - (C_{\mu}N^{a}_{\nu; a} + C_{\nu}N^{a}_{\mu; a}).\nonumber\end{split}\end{equation}

\hspace{-0.6 cm}Taking into account Lemma \ref{ct} {\bf (b)}, the above equation can be expressed in the form
\begin{equation}\begin{split}\label{EEEXy}\\[- 0.7 cm]0& = E_{(\mu\nu)} = g_{\mu\nu}H - 2H_{\mu\nu} - 2\alpha_{\mu\nu} - 2g_{\mu\nu}C +
\phi_{\mu\nu} + \theta_{\mu\nu} - \psi_{\mu\nu} + N^{a}_{\epsilon; a}(\Lambda_{\mu\nu}\!\,^{\epsilon} +  \Lambda_{\nu\mu}\!\,^{\epsilon})
\\& \ \ \ \ \ \ \ \ \ \ \ \ \ \
+ \, 2g_{\mu\nu}C^{\epsilon}N^{a}_{\epsilon; a} - (C_{\mu}N^{a}_{\nu; a} + C_{\nu}N^{a}_{\mu; a})\\[- 0.8 cm].\end{split}\end{equation}

\hspace{- 0.6 cm}As easily checked, the relation $\Lambda^{\alpha}_{\mu\nu} = \gamma^{\alpha}_{\mu\nu} - \gamma^{\alpha}_{\nu\mu}$ implies that
\begin{equation} \label{xtHT} H_{\mu\nu} = \sigma_{\mu\nu} - \varpi_{\mu\nu} + \omega_{\mu\nu}  - \alpha_{\mu\nu},\end{equation}
\begin{equation}\label{xtH} H := g^{\mu\nu}H_{\mu\nu} = \sigma - \varpi + \omega - \alpha.\end{equation}

\hspace{- 0.6 cm} Substituting (\ref{xtHT}) and (\ref{xtH}) in (\ref{EEEXy}), setting $N_{\beta}: = N^{a}_{\beta; a}$, we find that
\begin{equation}\begin{split}\\[- 0.7 cm]\label{xtE} 0 &= E_{(\mu\nu)}: = g_{\mu\nu}(\sigma - \varpi - \alpha + \omega) -
2(\sigma_{\mu\nu} - \varpi_{\mu\nu} + \omega_{\mu\nu}) - 2g_{\mu\nu} C\\& \ \ \ \ \ \ \ \ \ \ \ \ \ \, \ + (\theta_{\mu\nu} + \phi_{\mu\nu} - \psi_{\mu\nu}) +
N_{\beta}(\Lambda_{\mu\nu}\!\,^{\beta} + \Lambda_{\nu\mu}\!\,^{\beta}) +  2g_{\mu\nu}C^{\beta}N_{\beta}\\& \, \, \, \ \ \ \ \ \ \ \ \ \ \ \ \
- (C_{\mu}N_{\nu} + C_{\nu}N_{\mu}).\end{split}\end{equation}
\vspace{- 0.1 cm}On the other hand, by (\ref{symx}), we have
\begin{equation}\label{xtWT}\\[- 0.5 cm] \,\overcirc{R}_{(\mu\nu)} = - \frac{1}{2}\,(\theta_{\mu\nu} - \psi_{\mu\nu} + \phi_{\mu\nu}) +
\omega_{\mu\nu} + Q_{(\mu\nu)},\end{equation}
\\[ - 0.5 cm]\begin{equation}\label{xtscalar} \ \overcirc{\cal R} = - \frac{1}{2}\,(\theta - \psi + \phi) + \omega + Q; \ \ \ Q: = g^{\mu\nu}Q_{\mu\nu}.\end{equation}

\hspace{- 0.6 cm}Solving for $\omega_{\mu\nu}$ and $\omega$ in (\ref{xtWT}) and (\ref{xtscalar}) respectively and substituting in (\ref{xtE}), we obtain
\begin{equation}\begin{split}\label{ZXt} 0 &= E_{(\mu\nu)}: = g_{\mu\nu}(\sigma - \varpi - \alpha + \{ \, \overcirc{\cal R} - \frac{1}{2}\,(\psi - \phi - \theta) - Q\} - 2C) -
2(\sigma_{\mu\nu} - \varpi_{\mu\nu})\\
& \ \ \ \ \ \ \ \ \ \ \ \ \ \ \ - \ \{2 \ \overcirc{\cal R}_{(\mu\nu)} - (\psi_{\mu\nu} - \phi_{\mu\nu} - \theta_{\mu\nu}) - Q_{(\mu\nu)}\}
 + (\theta_{\mu\nu} + \phi_{\mu\nu} - \psi_{\mu\nu})\\
& \ \ \ \ \ \ \ \ \ \ \ \ \ \ \ +  \ N_{\beta}
(\Lambda_{\mu\nu}\!\,^{\beta} + \Lambda_{\nu\mu}\!\,^{\beta}) +  2g_{\mu\nu}C^{\beta}N_{\beta} -
(C_{\mu}N_{\nu} + C_{\nu}N_{\mu}).\end{split}\end{equation}

\hspace{- 0.8 cm} Let $N^{\beta} := g^{\beta\epsilon}N_{\epsilon}.$ Taking the relation $\frac{1}{2}(\phi - \psi + \theta) - \alpha - 2C = 0$ (Lemma \ref{ct} {\bf (a)})
into account, (\ref{ZXt}) reduces to
\begin{equation}\label{Tarek}\begin{split}  0 &= E_{(\mu\nu)}: = (g_{\mu\nu} \ \overcirc{\cal R}  - 2 \ \overcirc{R}_{(\mu\nu)})
+ g_{\mu\nu}(\sigma - \varpi - Q) - \ 2(\sigma_{\mu\nu} - \varpi_{\mu\nu} - Q_{(\mu\nu)})\\
& \ \ \ \ \ \ \ \ \ \ \ \ \ \ \ +  \ N^{\beta}
(\Lambda_{\mu\nu\beta} + \Lambda_{\nu\mu\beta}) +  2g_{\mu\nu}C^{\beta}N_{\beta} -
(C_{\mu}N_{\nu} + C_{\nu}N_{\mu}),\end{split}\end{equation}
which represents the {\bf symmetric part} of the horizontal unified field equations (\ref{EEEXz}) expressed in terms of the fundamental tensors of Table 1.

\bigskip

\hspace{- 0.6 cm}Finally, setting
\begin{equation}\label{newtensors}M_{\mu\nu} := N^{\beta}\Lambda_{\mu\nu\beta}, \ \ \ \ \ Z_{\mu\nu} := C_{\mu}N_{\nu}, \ \ \ \
Z := g^{\mu\nu}Z_{\mu\nu},\end{equation}
equation (\ref{Tarek}) can be written in the more informative form:
\begin{equation}\label{first order}\overcirc{R}_{(\mu\nu)} - \frac{1}{2}\, g_{\mu\nu}\,\,\overcirc{\cal R} = T_{(\mu\nu)};\end{equation}
\begin{equation}\label{ems}T_{(\mu\nu)} := \frac{1}{2}\,g_{\mu\nu}(\sigma - \varpi - Q + 2Z)   - (\sigma_{\mu\nu} - \varpi_{\mu\nu} - Q_{(\mu\nu)}) +
(\frac{1}{2}\,N_{\beta}\Omega^{\beta}_{\mu\nu} - Z_{(\mu\nu)}).\end{equation}
\hspace{0.6 cm} According to (\ref{first order}), $T_{(\mu\nu)}$ may be interpreted as the
{\bf horizontal geometric energy-momentum tensor} (as will be clear in section 7), constructed from the
symmetric tensors of Table 1, together with $N_{\beta}\Omega^{\beta}_{\mu\nu}$, $Q_{(\mu\nu)}$ and $Z_{(\mu\nu)}$. On the other hand,
in view of (\ref{first order}), (\ref{ems}) and the fact that
\, $\overcirc{R}_{[\mu\nu]} = \frac{1}{2}\, \mathfrak{S}_{\mu, \nu, \alpha} \ \overcirc{C}^{\alpha}_{\mu a} R^{a}_{\nu \alpha}$ (by (\ref{skewz})), the
horizontal Einstein tensor (\ref{HE}) takes the form
\begin{equation}\begin{split}\label{cxa} \overcirc{G}_{\mu\nu} & = \ \overcirc{R}_{\mu\nu} - \frac{1}{2}\, g_{\mu\nu} \ \overcirc{\cal R}\\
& \ = \, \frac{1}{2}\, g_{\mu\nu}(\sigma - \varpi) + (\varpi_{\mu\nu} - \sigma_{\mu\nu}) + \frac{1}{2}\,g_{\mu\nu}(2Z - Q) + \frac{1}{2} \, N_{\beta}\,
\Omega^{\beta}_{\mu\nu} -
Z_{(\mu\nu)} + Q_{(\mu\nu)}\\
& \ \ \, \ \ + \ \frac{1}{2}\, \mathfrak{S}_{\mu, \nu, \alpha} \ \overcirc{C}^{\alpha}_{\mu a} R^{a}_{\nu \alpha},\end{split}\end{equation}
which is, by (\ref{HES}), subject to the identity 
\begin{equation}\label{HESxx} \,\overcirc{G}^{\mu}\!\,_{\sigma{o\atop|}\mu} = R^{a}_{\sigma\mu}\,\,\overcirc{P}^{\mu}_{a} +
\frac{1}{2}\,R^{a}_{\alpha\mu}\,\,\overcirc{P}^{\alpha\mu}\!\,_{\sigma a}.\end{equation}

\bigskip

Now, we consider the {\bf skew-symmetric part} of the horizontal field equations (\ref{EEEXz}). By Lemma \ref{ct} {\bf (c)} and (\ref{newtensors}), we obtain
\begin{equation}\label{EEEXyt}0 = E_{[\mu\nu]} = 2\gamma_{\mu\nu} + \eta_{\mu\nu} - \epsilon_{\mu\nu}
- 2\xi_{\mu\nu} - \chi_{\mu\nu} + 2N_{\beta}\Lambda^{\beta}_{\mu\nu} + 2(M_{[\mu\nu]} - Z_{[\mu\nu]}) +
2\mathfrak{S}_{\mu, \nu, \epsilon}C^{\epsilon}_{\nu a}R^{a}_{\epsilon\mu}\end{equation}

\hspace{-0.6 cm}Expressed in terms of the fundamental tensors of Table 1, relation (\ref{h}) is given by
\begin{equation}\label{tFBII}\eta_{\mu\nu} + \epsilon_{\mu\nu} - \chi_{\mu\nu} =  \mathfrak{S}_{\mu, \nu, \epsilon}C^{\epsilon}_{\nu a}R^{a}_{\epsilon\mu},\end{equation}

\hspace{- 0.6 cm}which, when inserted in (\ref{EEEXyt}), yields
\begin{equation}\label{tSM} 0 = E_{[\mu\nu]} = 2\{(\gamma_{\mu\nu} - \epsilon_{\mu\nu} - \xi_{\mu\nu} +
N_{\beta}\Lambda^{\beta}_{\mu\nu}) + (M_{[\mu\nu]} - Z_{[\mu\nu]})\} +  3\mathfrak{S}_{\mu, \nu, \epsilon}C^{\epsilon}_{\nu a}R^{a}_{\epsilon\mu}.\end{equation}

\hspace{- 0.6 cm}Moreover, since
\begin{equation}\epsilon_{\mu\nu} = C_{\mu|\nu} - C_{\nu|\mu} = (\delta_{\nu}C_{\mu} - \delta_{\mu}C_{\nu}) - \eta_{\mu\nu},\nonumber\end{equation}
\hspace{0.02 cm}it follows, by (\ref{tSM}), that
\begin{equation}\label{tMax}\delta_{\nu}C_{\mu} - \delta_{\mu}C_{\nu} = (\gamma_{\mu\nu} - \xi_{\mu\nu} + \eta_{\mu\nu} +
N_{\beta}\Lambda^{\beta}_{\mu\nu}) + (M_{[\mu\nu]} - Z_{[\mu\nu]}) +
\frac{3}{2}\mathfrak{S}_{\mu, \nu, \epsilon}C^{\epsilon}_{\nu a}R^{a}_{\epsilon\mu}.\end{equation}
The above equation can be written in the more informative form
\begin{equation}\label{CUzx}F_{\mu\nu} = \delta_{\nu}C_{\mu} - \delta_{\mu}C_{\nu},\end{equation}
where
\begin{equation}\begin{split}\label{tEMF}F_{\mu\nu}: &=  (\gamma_{\mu\nu} - \xi_{\mu\nu} + \eta_{\mu\nu} + N_{\beta}\Lambda^{\beta}_{\mu\nu})
+ (M_{[\mu\nu]} - Z_{[\mu\nu]}) + \frac{3}{2}\,\mathfrak{S}_{\mu, \nu, \epsilon}C^{\epsilon}_{\nu a}R^{a}_{\epsilon\mu} \\
& = (\gamma_{\mu\nu} - \xi_{\mu\nu} + \eta_{\mu\nu}) + N^{\beta}(\gamma_{\mu\nu\beta} + \Lambda_{\beta\mu\nu}) +
(\frac{1}{2}\,N^{\beta}\Lambda_{\beta\mu\nu} - Z_{[\mu\nu]})
+ \frac{3}{2}\,\mathfrak{S}_{\mu, \nu, \epsilon}C^{\epsilon}_{\nu a}R^{a}_{\epsilon\mu}\end{split}\end{equation}
Now, recalling that $[\delta_{\mu}, \delta_{\nu}] = R^{a}_{\mu\nu}\,\dot{\partial_a}$, which indicates the non-commutativity of the
operator $\delta_{\mu}$, and
\begin{equation}F_{\mu\nu{o\atop{|}}\sigma} + F_{\nu\sigma{o\atop{|}}\mu} + F_{\sigma\mu{o\atop{|}}\nu} =
\delta_{\sigma}F_{\mu\nu} + \delta_{\mu}F_{\nu\sigma} + \delta_{\nu}F_{\sigma\mu},\\[0.1 cm]\nonumber\end{equation}
one can write, using (\ref{CUzx}), the identity
\vspace{0.2 cm}\begin{equation}\label{GMESz}\mathfrak{S}_{\mu, \nu, \sigma} \, F_{\mu\nu{o\atop|}\sigma} = - \,\mathfrak{S}_{\mu, \nu, \sigma} \, R^{a}_{\mu\nu}
\dot{\partial_a}C_{\sigma}.\end{equation}

According to (\ref{CUzx}), considering its right hand side as a generalization of the curl of the horizontal basic vector $C_{\mu}$
in the context of the
EAP-geometry, then the tensor $F_{\mu\nu}$ can be considered as the {\bf horizontal geometric electromagnetic field strength}, the vector
$C_{\mu}$ as the {\bf horizontal geometric electromagnetic
potential} and equation (\ref{CUzx}) as a {\bf generalized} form of the {\bf horizontal Maxwell's equations}. Again, by (\ref{tEMF}), $F_{\mu\nu}$ is
constructed from the horizontal skew-symmetric fundamental tensors of the EAP-space (Table 1) together with the skew-symmetric tensors
$N^{\beta}\gamma_{\mu\nu\beta}$, $N^{\beta}\Lambda_{\beta\mu\nu}$,  $\mathfrak{S}_{\mu, \nu, \epsilon}C^{\epsilon}_{\nu a}R^{a}_{\epsilon\mu}$
and $Z_{[\mu\nu]}$. It is thus constructed from a purely geometric standpoint.

\bigskip

Now, let \begin{equation}J^{\mu} := F^{\mu\nu}\!\,_{{o\atop{|}}\nu}.\nonumber\end{equation}
Then, by the commutation formula
\begin{equation}\label{kj}F^{\mu\nu}\!\,_{{o\atop{|}}\alpha\beta} - F^{\mu\nu}\!\,_{{o\atop{|}}\beta\alpha} =
F^{\mu\epsilon}\,\overcirc{R}^{\nu}_{\epsilon\beta\alpha} +
 F^{\epsilon\nu}\,\overcirc{R}^{\mu}_{\epsilon\beta\alpha} + R^{a}_{\beta\alpha}F^{\mu\nu}\!\,_{{o\atop{||}}a},\nonumber\end{equation}
noting that $F^{\mu\nu}$ is skew-symmetric, we conclude that
\begin{equation}\label{GPMx}2\,F^{\mu\nu}\!\,_{{o\atop{|}}\mu\nu} = F^{\mu\nu}\!\,_{{o\atop{|}}\mu\nu} - F^{\mu\nu}\!\,_{{o\atop{|}}\nu\mu} =
- F^{\mu\epsilon}\,\overcirc{R}_{\epsilon\mu} +  F^{\epsilon\nu}\,\overcirc{R}_{\epsilon\nu} - R^{a}_{\mu\nu}F^{\mu\nu}\!\,_{{o\atop{||}}a}.\nonumber\end{equation}
Consequently,
\begin{equation}\label{CONSz}J^{\mu}\!\,_{{o\atop{|}}\mu} = \frac{1}{2}\,\{ F^{\epsilon\mu}(\,\overcirc{R}_{\mu\epsilon} - \,\overcirc{R}_{\epsilon\mu}) +
R^{a}_{\mu\nu}F^{\mu\nu}\!\,_{{o\atop{||}}a}\}.\end{equation}

\bigskip

This relation will be discussed in subsection 6.1.

\bigskip

Finally, setting
\begin{equation}\label{Ax}{\cal V}_{\mu\nu}:= N^{\beta}\gamma_{\beta\mu\nu} - Z_{\mu\nu} + Q_{\mu\nu},\nonumber\end{equation}

\begin{equation}\label{Ay}{\cal U}_{\mu\nu} : = \frac{1}{2}\,\mathfrak{S}_{\mu, \nu, \alpha} \ \overcirc{C}^{\alpha}_{\mu a} R^{a}_{\nu \alpha},\nonumber\end{equation}
we find, by (\ref{cxa}), (\ref{tEMF}) and
\begin{equation}\mathfrak{S}_{\mu, \nu, \alpha}\,{C}^{\alpha}_{\mu a} R^{a}_{\nu \alpha} =
\mathfrak{S}_{\mu, \nu, \alpha} \ \overcirc{C}^{\alpha}_{\mu a} R^{a}_{\nu \alpha} + 2 Q_{[\mu\nu]},\nonumber\end{equation}
that
\begin{equation}\label{az}\overcirc{G}_{\mu\nu} := \frac{1}{2}\,g_{\mu\nu}(\sigma - h) + (h_{\mu\nu} - \sigma_{\mu\nu}) + \frac{1}{2}\,g_{\mu\nu}(2Z - Q) +
{\cal V}_{(\mu\nu)} + {\cal U}_{\mu\nu},\end{equation}
\begin{equation}\label{bz}F_{\mu\nu} = (\gamma_{\mu\nu} - \xi_{\mu\nu} + \eta_{\mu\nu}) +
N^{\beta}(\gamma_{\mu\nu\beta} + \Lambda_{\beta\mu\nu}) + {\cal V}_{[\mu\nu]} + {\cal U}_{\mu\nu} +
{\cal W}_{\mu\nu};\end{equation}

\begin{equation}{\cal W}_{\mu\nu}: = \mathfrak{S}_{\mu, \nu, \alpha}\,{C}^{\alpha}_{\mu a} R^{a}_{\nu \alpha}.\nonumber\\[0.2 cm]\end{equation}

It is clear, by (\ref{az}) and (\ref{bz}), that the skew-symmetric tensor ${\cal U}_{\mu\nu}$ contributes to both gravitational and electromagnetic effects.
Moreover, the
tensor ${\cal V}_{\mu\nu}$ could be interpreted as representing a kind of {\it mutual interaction} between gravity and electromagnetism,
since its symmetric (resp. skew-symmetric) part gives rise to gravitational (resp. electromagnetic) effects.


\subsection{Splitting the vertical field equations}

\hspace{0.5 cm}Considering the {\bf symmetric part} of the vertical field equations (\ref{xax}), taking into consideration the vertical analogue of Lemma \ref{ct} {\bf (b)}, we obtain, similar to (\ref{EEEXy}),
\begin{equation}\label{bbx} 0 = E_{(ab)}: = g_{ab} V - 2V_{ab} - 2g_{ab} (C^{e}\!\,_{||e} - C^{e}C_{e}) -
2\alpha_{ab} + \theta_{ab} + \phi_{ab} - \psi_{ab}.\end{equation}

Proceeding as we did in the derivation of the symmetric part of the horizontal field equations, taking into account (\ref{VSM}) and the vertical
analogue of Lemma \ref{ct}{\bf (a)} (setting $\bar{\sigma} = {trace}(\sigma_{ab})$ and $\bar{\varpi} = {trace}(\varpi_{ab})$),
we finally arrive at the relation

\begin{equation} 0 = E_{(ab)}: = (g_{ab} \ \overcirc{\cal S}  - 2 \ \overcirc{S}_{ab}) + g_{ab}(\bar{\sigma} - \bar{\varpi}) -
2(\sigma_{ab} - \varpi_{ab}),\nonumber\end{equation}
which can be written in the more informative form:
\begin{equation}\label{SYMx}\ \overcirc{S}_{ab} - \frac{1}{2} \, g_{ab} \ \overcirc{\cal S} = T_{ab};\end{equation}
\begin{equation}\label{xems}T_{ab}: = \frac{1}{2}\,g_{ab}(\bar{\sigma} - \bar{\varpi})  - (\sigma_{ab} - \varpi_{ab}).\end{equation}
Moreover, (\ref{SYMx}) implies, using (\ref{vee}), that
\begin{equation}\label{cons}T^{a}\,_{b{o\atop{||}} a} = 0.\end{equation}
Consequently, in view of (\ref{SYMx}) and (\ref{cons}), $T_{ab}$ could be interpreted as the {\bf vertical geometric
energy-momentum tensor} for both matter and electromagnetism, which is, according to (\ref{xems}), constructed from the vertical
symmetric fundamental tensors of the EAP-space (Table 1).

\bigskip
Considering the {\bf skew-symmetric part} of (\ref{xax}), following the
same steps as in the previous section, with necessary changes, we
conclude that
\begin{equation}\label{po}F_{ab} = \dot{\partial_b}C_{a} - \dot{\partial_a}C_{b},\end{equation}
where
\begin{equation}\label{MAXTx}F_{ab} :=  \gamma_{ab} - \xi_{ab} + \eta_{ab}\end{equation}
is the {\bf vertical geometric electromagnetic field} tensor. Moreover, (\ref{po})
represents the {\bf vertical generalized Maxwell's equations} in which $F_{ab}$ is expressed as the generalized curl of the vertical basic
vector $C_{a}$. Consequently, $C_{a}$ may be interpreted as the {\bf vertical geometric}
{\bf electromagnetic potential}. Also, by (\ref{po}) and the relation
$$F_{ab{o\atop{||}}c} + F_{bc{o\atop{||}}a} + F_{ca{o\atop{||}}b} =
\dot{\partial_c}F_{ab} + \dot{\partial_a}F_{bc} + \dot{\partial_b}F_{ca},$$
we obtain the identity

\begin{equation}\label{GMEz}\mathfrak{S}_{a, b, c} \, F_{ab{o\atop||}c} = 0.\end{equation}

It is clear, by (\ref{MAXTx}), that $F_{ab}$ is constructed from the vertical skew-symmetric tensors of Table 1.

\bigskip

Finally, if we set
\begin{equation}J^{a} := F^{ab}\!\,_{{o\atop{|}}b}\nonumber\end{equation}
then,  similar to (\ref{CONSz}), $J^{a}$ satisfies the identity 
\begin{equation}\label{hala}J^{a}\!\,_{{o\atop{|}}a} = 0.\end{equation}
Hence, $J^{a}$ represents the {\bf vertical geometric current density} and (\ref{hala}) represents a generalization of the
conservation law of the current density.

\bigskip

We note that the equations (\ref{MAXTx}) to (\ref{hala}) are formally similar to those obtained in the GFT \cite{aaa},
with each geometric object of the GFT replaced by the corresponding vertical geometric object in the EAP-context.

\bigskip

To sum up, in sections 4 and 5, we have constructed a unified field
theory representing a natural generalization of the GFT in which the
unique metric $d$-connection defined in Theorem \ref{metric} plays
the role of the Riemannian connection in the GFT. The constructed
theory has given  rise to {\bf two sets of field equations}
(\ref{cxa}), (\ref{CUzx}) and (\ref{SYMx}), (\ref{po}) subject to
the identities (\ref{GMESz}), (\ref{CONSz}), (\ref{GMEz}) and the
conservation laws (\ref{cons}) and (\ref{hala}). Certain geometrical
objects derived from the fundamental vector fields have been
identified with the {\bf material energy}, {\bf electromagnetic
field}, {\bf electromagnetic potential} and the {\bf current
density}. The dependence of the geometric objects considered on the
directional argument $y$, in addition to the positional argument
$x$, has made our constructed field theory wider in scope and richer
in content than the GFT.
\end{mysect}


\begin{mysect}{Important Special Cases}

We now investigate some interesting special cases for the constructed field equations by imposing extra conditions which are natural
on either the nonlinear connection or
the canonical $d$-connection or both. These conditions illuminate our understanding of the physical contents of the constructed theory.
We consider the following cases.
\vspace{- 0.5 cm}


\subsection{Integrability condition}
\hspace{0.45 cm} We investigate the form of the horizontal field equations obtained under the additional assumption
that {the nonlinear connection is integrable}, that is, $R^{a}_{\mu\nu}$ vanishes.
We refer to this condition as the integrability condition. No extra conditions are imposed on the canonical $d$-connection.

\bigskip

Let $R^{a}_{\mu\nu} = 0$. Then, noting in this case that $T_{\mu\nu} = T_{(\mu\nu)}$ and setting $A_{\mu\nu}:=
N^{\beta}\gamma_{\beta\mu\nu} - Z_{\mu\nu}$, equations (\ref{ems}), (\ref{tEMF}) and (\ref{CUzx}) take the form

\begin{equation}\label{cazx}T_{\mu\nu} = \{\frac{1}{2}\,g_{\mu\nu}(\sigma - \varpi) + (\varpi_{\mu\nu} - \sigma_{\mu\nu})\} + g_{\mu\nu}Z +
A_{(\mu\nu)},\end{equation}
\begin{equation}\label{cbzx}F_{\mu\nu} = (\gamma_{\mu\nu} - \xi_{\mu\nu} + \eta_{\mu\nu}) +
N^{\beta}(\gamma_{\mu\nu\beta} + \Lambda_{\beta\mu\nu}) + A_{[\mu\nu]},\vspace{- 0.3 cm}\end{equation}

\begin{equation}\label{gage}F_{\mu\nu} = \delta_{\nu}C_{\mu} - \delta_{\mu}C_{\nu}.\end{equation}
The terms between curly brackes in (\ref{cazx}) are similar in form to those obtained in the context of the GFT.
The gauge invariance of (\ref{gage}) will be discussed in the concluding remarks.

\bigskip

Again, by (\ref{cazx}) and (\ref{cbzx}), $A_{\mu\nu}$ could be interpreted as representing a kind of {\it mutual interaction} between gravity and electromagnetism,
since its symmetric (resp. skew-symmetric) part gives rise to gravitational (resp. electromagnetic) effects.

\bigskip

It is clear that $T_{\mu\nu}$ is symmetric. Moreover, by relation ({\ref {HESxx}}), noting that $G_{\mu\nu} = T_{\mu\nu}$, the energy-momentum tensor
$T_{\mu\nu}$ satisfies the {\bf conservation law}
\begin{equation}\label{cssx}T^{\mu}\,_{\nu{o\atop{|}}\mu} = 0.\end{equation}
\par Also, by (\ref{GMESz}), $F_{\mu\nu}$ satisfies the identity
\begin{equation}\label{MESz}\mathfrak{S}_{\mu, \nu, \sigma} \, F_{\mu\nu{o\atop|}\sigma} =  0.\end{equation}
\par
Finally, in view of (\ref{CONSz}), noting that \, $\overcirc{R}_{\mu\nu}$, in this case, is symmetric (by (\ref{skewz})), the
geometric current density satisfies the {\bf conservation law}
\begin{equation}\label{ssx}J^{\mu}\!\,_{{o\atop{|}}\mu} = 0.\end{equation}
\par
On the other hand, the vertical field equations remain the same as those obtained in section 5.


\subsection{Cartan-type case}

\hspace{0.5 cm} A $d$-connection $D = (\Gamma^{\alpha}_{\mu\nu}, \, \Gamma^{a}_{b\mu}, \, C^{\alpha}_{\mu c}, \, C^{a}_{bc})$ on $TM$
is said to be of Cartan-type  if
\begin{equation}y^{a}\!\,_{|\mu} = 0, \ \ \ \ \ \ \ y^{a}\!\,_{||c} = \delta^{a}_{c}.\nonumber\end{equation}

Let $(TM, \lambda)$ be an EAP-space.
Assume that the canonical $d$-connection $D$ is of Cartan-type. Then we have \cite{EAP}:
\begin{description}
\item [(a)] The nonlinear connection $N^{a}_{\mu}$ is expressed in the form
$N^{a}_{\mu} =  y^{b}( \, \undersym{\lambda}{i}^{a}\partial_{\mu} \ \undersym{\lambda}{i}_{b})$.\footnote{A similar expression is found in \cite{WW},
but in a completely different situation.}
\item [(b)] The nonlinear connection $N^{a}_{\mu}$ is integrable: $R^{a}_{\mu\nu} = 0$.
\item [(c)] $P^{a}_{\mu c} = T^{a}_{bc} = 0.$ Consequently, $C^{a}_{bc}$ is symmetric,
$\gamma^{a}_{bc} = 0$ and $C^{a}_{bc} = \ \overcirc{C}^{a}_{bc}$.
 \item [(d)] $\dot{\partial_b}N^{a}_{\mu} = \Gamma^{a}_{b\mu}$ and $N^{a}_{\mu}$ is homogeneous. Consequently,
$\Gamma^{a}_{b\mu}$ is positively homogeneous of degree $0$ in $y$.
\end{description}

Assume that $D$ is of Cartan-type. Then the horizontal Einstein tensor satisfies the identity
\begin{equation}\label{HEE}\,\,\overcirc{G}^{\mu}\!\,_{\sigma{o\atop|}\mu} = 0.\nonumber\end{equation}
Moreover, all vertical second rank tensors in Table 1 vanish identically.
\bigskip

We now consider the case dealt with in section 4 under the
additional assumption that the canonical $d$-connection is of {Cartan-type}. By {\bf (b)} above, the nonlinear connection
$N^{\alpha}_{\mu}$ is integrable. Consequently, the Cartan-type case
can be regarded as a special case of the Integrability case,
obtained by setting $N^{a}_{\mu} =
y^{b}(\,\undersym{\lambda}{i}^{a}\partial_{\mu}
\,\undersym{\lambda}{i}_{b})$ and $R^{a}_{\mu\nu} = 0$ (among other
things). Accordingly, relations (\ref{cazx}) to (\ref{ssx}) remain
valid under the Cartan-type condition.  In particular, we have
\begin{equation}\label{gages}F_{\mu\nu} = \delta_{\nu}C_{\mu} - \delta_{\mu}C_{\nu}.\end{equation}

The gauge invariance of (\ref{gages}) will be discussed in the concluding remarks.

\bigskip

On the other hand, there are
no vertical field equations (all vertical objects of Table 1
vanish). The advantage in this case is that the nonlinear
connection, consequently, all geometric objects considered, are
expressed explicitly in terms of the fundamental vector fields
$\lambda$'s.  It is for this reason that the horizontal field
equations in the Cartan-type case lend themselves to the process of
linearization.


\subsection{Berwald-type case}

\hspace{0.5 cm} A $d$-connection $D = (\Gamma^{\alpha}_{\mu\nu}, \, \Gamma^{a}_{b\mu}, \, C^{\alpha}_{\mu c}, \, C^{a}_{bc})$ on $TM$
is said to be of Berwald-type if
\begin{equation}\label{Berwald} \ \ \dot{\partial_{b}}N^{a}_{\mu} = \Gamma^{a}_{b\mu}; \ \ \ \ \ \ \ \ \ C^{\alpha}_{\mu c} = 0.\nonumber\end{equation}

Assume that $D$ is of Berwald-type. Then \cite{EAP}
\begin{description}

\item [(a)] ${\lambda}_{\mu}$ are functions of the positional argument $x$ only.
Consequently, so are $g_{\mu\nu}$.
\item [(b)] Both $\Lambda^{\alpha}_{\mu\nu}$ and $\gamma^{\alpha}_{\mu\nu}$ are functions of the positional argument
$x$ only. Consequently, so are $C_{\mu}$.
\item [(c)] $ \ \overcirc{C}^{\alpha}_{\mu c} = 0$. Consequently, $\gamma^{\alpha}_{\mu c} = 0$ so that \\[- 0.1 cm]
\begin{equation}\label{SYM}\overcirc{R}_{\beta\mu} = \, \overcirc{R}_{(\beta\mu)} = - \frac{1}{2}\,(\theta_{\beta\mu} - \psi_{\beta\mu} + \phi_{\beta\mu}) +
\omega_{\beta\mu}.\\[- 0.2 cm]\end{equation}
\item [(d)] $\overcirc{P}^{\alpha}_{\beta\nu  c} = 0.$ Consequently,
the horizontal Einstein tensor satisfies the identity
\begin{equation}\label{HEE}\,\,\overcirc{G}^{\mu}\!\,_{\sigma{o\atop|}\mu} = 0.\nonumber\end{equation}
\end{description}

We now assume that the canonical $d$-connection is of {Berwald-type}. In this case, {\it no conditions
are imposed on the nonlinear connection.}   In other words, the Berwald-type case is not deducible from the case dealt with in section 4.
We will therefore derive the horizontal field equations from scratch. Since the horizontal fundamental vector fields
and all horizontal geometric objects involved are functions of the positional argument $x$ only, the Euler-Lagrange equations have the form
\begin{equation}\label{ELEs}\frac{\delta {\cal H}}{\delta \lambda_{\beta}}: = \frac{\partial {\cal H}}{\partial \lambda_{\beta}} -
\frac{\partial}{\partial x^{\gamma}}\bigg(\frac{\partial {\cal H}}
{\partial \lambda_{\beta, \gamma}}\bigg) = 0.\nonumber\end{equation}
Moreover, by the fact that $C^{\alpha}_{\mu c} = 0$,
the {\bf horizontal field equations} in this case, as can be checked, are given by
\begin{equation}\begin{split}\label{BER}0 &= E_{\mu\nu} := g_{\mu\nu}H - 2H_{\mu\nu} -
2(C_{\mu}C_{\nu} - C_{\nu|\mu}) + 2g_{\mu\nu}(C^{\epsilon}C_{\epsilon} - C^{\epsilon}\!\,_{|\epsilon})\\&
\ \ \ \ \ \ \ \ \ \ \ \ \ - \, 2(C^{\epsilon}\Lambda_{\mu\epsilon\nu} +
g^{\epsilon\alpha}\Lambda_{\mu\nu\alpha|\epsilon}),\end{split}\end{equation}
which are {\bf similar in form} to the field equations of the GFT. Moreover, all geometrical
objects involved in (\ref{BER}) are functions of the positional argument $x$ only.

\bigskip

Proceeding as we did with the symmetric part of the field equations, taking into account
(\ref{SYM}) and (\ref{HEE}), we deduce that

\begin{equation}\label{order}\overcirc{R}_{\mu\nu} - \frac{1}{2}\, g_{\mu\nu}\,\,\overcirc{\cal R} = T_{\mu\nu};\\[- 0.3 cm]\nonumber\end{equation}

\begin{equation}\label{MT}T_{\mu\nu}: = \frac{1}{2}\,g_{\mu\nu}(\sigma - \varpi)   - (\sigma_{\mu\nu} - \varpi_{\mu\nu}),\end{equation}
where the energy-momentum tensor $T_{\mu\nu}$ is subject to the {\bf conservation law}
\begin{equation}\label{sss}T^{\mu}\,_{\nu{o\atop{|}}\mu} = 0.\end{equation}

Again similar to the skew-symmetric part of the field equations, noting that
\linebreak $[\delta_{\mu}, \delta_{\nu}] = R^{a}_{\mu\nu}\dot{\partial_a}$ and
$C_{\sigma} = C_{\sigma}(x)$, we obtain
the identity
\begin{equation}\label{GMESx}\mathfrak{S}_{\mu, \nu, \sigma} \, F_{\mu\nu{o\atop|}\sigma} = - \,\mathfrak{S}_{\mu, \nu, \sigma} \, R^{a}_{\mu\nu}
\dot{\partial_a}C_{\sigma} = 0,\end{equation}
where the electromagenetic field is given by:
\begin{equation}\label{xyzc}F_{\mu\nu}: =  \gamma_{\mu\nu} - \xi_{\mu\nu} + \eta_{\mu\nu}.\end{equation}

Moreover, $F_{\mu\nu}$ is expressed as the curl of the horizontal basic vector, namely,
\begin{equation}\label{gagex}F_{\mu\nu} = \partial_{\nu}C_{\mu} - \partial_{\mu}C_{\nu}.\end{equation}

The gauge invariance of (\ref{gagex}) will be discussed in the concluding remarks.

\bigskip
Finally, if
\begin{equation}J^{\mu} := F^{\mu\nu}\!\,_{{o\atop{|}}\nu},\end{equation}
then, by (\ref{CONSz}), noting that \,$\overcirc{R}_{\mu\nu}$ is symmetric and $F^{\mu\nu}\!\,_{{o\atop{||}}a} = 0$, we conclude that
$J^{\mu}$ satisfies the {\bf conservation law}
\begin{equation}\label{wwwz}J^{\mu}\!\,_{{o\atop{|}}\mu} = 0.\end{equation}

In the Berwald-type case, equations (\ref{order}) to (\ref{wwwz})
are identical in form to those obtained in the context of the GFT.
Moreover, all geometric objects involved are functions of the
positional argument $x$ only. Consequently, the horizontal field
equations obtained under the Berwald-type condition are actually an
{exact replica} of the GFT. Nevertheless, the vertical field
equations under the Berwald-type condition are still alive and are the same as those
obtained in the general case (\ref{xax}). This means that, in this case, we have the GFT plus something else whose
physical essence is not yet revealed.


\subsection{Recovering the GFT (The Cartan-Berwald case)}

\hspace{0.6 cm} We finally assume that the canonical $d$-connection is both of Berwald- and Cartan-type. Then we have, in this case, \cite{EAP}

\begin{description}
\item [(a)] The $hh$-coefficients of the natural metric and the canonical $d$-connections are functions of the
positional argument $x$ only and are both identical to the coefficients of the corresponding
connections in the conventional AP-space.
\item [(b)] The torsion and the contortion of the EAP-space are functions
of the positional argument $x$ only and are given by
$${\bf T} = (\Lambda^{\alpha}_{\mu\nu}, \, 0, \, 0, \, 0 , \, 0); \ \ \ \ {\bf C} =
(\gamma^{\alpha}_{\mu\nu}, \, 0, \, 0, \, 0)$$
\item [(c)] The horizontal fundamental tensors of Table 1 are functions of the
positional argument $x$ only and are identical to their
corresponding tensors in the conventional AP-space (cf. \cite{FI}).

\end{description}

In this case, the horizontal field equations are given by (\ref{BER}), whereas the vertical field equations clearly disappear.
Moreover, relations (\ref{MT}) to (\ref{wwwz}) hold. 
Consequently, we have only one set of field equations which actually {\bf coincides} with those of the GFT.
This is the typical case in which the GFT is naturally retrieved.

\end{mysect}


\begin{mysect}{Discussion and Concluding Remarks}

In the next 2 tables, we summarize the important results obtained so far. We shall use the following abbreviations:

\bigskip
{\bf I}-condition, for the Integrability condition.

{\bf C}-condition, for the Cartan case.

{\bf B}-condition, for the Berwald case.

{\bf CB}-condition, for the Cartan-Berwald case.

{\bf MEC}, for the matter-energy conservation law.

{\bf CDC}, for the current density conservation law.

\bigskip

We also set
$${\mathbb{E}}_{\mu\nu} := \frac{1}{2}\,g_{\mu\nu}(\sigma - \varpi) +
(\varpi_{\mu\nu} - \sigma_{\mu\nu}), \ \ \ {\mathbb{F}}_{\mu\nu} := \gamma_{\mu\nu} - \xi_{\mu\nu} + \eta_{\mu\nu}, \\[- 0.3 cm]$$
$${\mathbb{E}}_{ab} := \frac{1}{2}\,g_{ab}(\bar{\sigma} - \bar{\varpi}) +
(\varpi_{ab} - \sigma_{ab}), \ \ \ {\mathbb{F}}_{ab} := \gamma_{ab} - \xi_{ab} + \eta_{ab}.$$

We note that ${\mathbb{E}}_{\mu\nu}$ and ${\mathbb{F}}_{\mu\nu}$ represent the Einstein and electromagnetic tensors in the context of the
GFT respectively.

\bigskip

{\it The next two tables should be considered together as a single entity since they are
complementary.} For example, the horizontal counterparts of the {\bf
CB}- and {\bf B}-conditions are identical. However, their vertical
counterparts are categorically different; the vertical counterparts
of the {\bf CB}-condition disappear, whereas those of the {\bf
B}-condition coincide with those of the general case. Similarly, the
vertical counterparts of the {\bf I}-condition and the general case
coincide, while their horizontal counterparts are again different. {\it It is only by considering the two tables as {\bf one unit} that the
similarity and the difference of the cases dealt with are revealed.}

\bigskip

On the other hand, though the horizontal conservation laws (Table 2) are similar in form in all cases considered (apart from
the general case), the geometric objects involved in the {\bf CB}-
and {\bf B}-conditions are functions of the positional argument $x$ only,
while in the {\bf C}- and {\bf I}- conditions they are function of
both the positional argument $x$ and the directional argument $y$.

\newpage

\begin{landscape}
\begin{center}
Table 2\\
{\bf Horizontal section}
\end{center}
\begin{tabular}{|l|l|l|l|l|l|}
 \hline
\multirow{2}{*}{\bf Condition}&\multicolumn{2}{l|}{\hspace{2.8
cm}{\bf Gravitation and Electromagnetic tensors}
}&\multicolumn{2}{l|}{\hspace{0.2 cm}
{\bf Conservation laws}}&\multirow{2}{*}{\bf Max. Equations}\\
\cline{2-5} \cline{2-5}  & \hspace{1.3 cm}{\bf Gravitation  tensor} & \hspace{1.1 cm}{\bf Electromagnetic tensor} &\hspace{0.5 cm}{\bf MEC}
&\hspace{0.2 cm}{\bf CDC} &\\
\hline {\bf CB-Condition} &$\,\overcirc{G}_{\mu\nu} = {\mathbb {E}}_{\mu\nu}$
& $F_{\mu\nu} = {\mathbb{F}}_{\mu\nu} $
&{${T}^{\mu}\,_{\sigma{o\atop{|}}\mu} = 0$}
&$J^{\mu}\!\,_{{o\atop|}\mu} = 0$
 & $F_{\mu\nu} = \partial_{\nu}C_{\mu} - \partial_{\mu}C_{\nu}$\\
\hline
 {\bf B-Condition} & $\,\overcirc{G}_{\mu\nu} = {\mathbb{E}}_{\mu\nu}$  &
 $F_{\mu\nu} = {\mathbb{F}}_{\mu\nu}$ &{${T}^{\mu}\,_{\sigma{o\atop{|}}\mu} = 0$}
 & $J^{\mu}\!\,_{{o\atop|}\mu} = 0$ &$F_{\mu\nu} = \partial_{\nu}C_{\mu} - \partial_{\mu}C_{\nu}$\\
\hline
{\bf C-Condition} & {$\,\overcirc{G}_{\mu\nu} = {\mathbb{E}}_{\mu\nu} + g_{\mu\nu}Z + A_{(\mu\nu)}$} &$F_{\mu\nu} = {\mathbb{F}}_{\mu\nu} +
N^{\beta}(\gamma_{\mu\nu\beta} +
\Lambda_{\beta\mu\nu}) + A_{[\mu\nu]}$ &
{${T}^{\mu}\,_{\sigma{o\atop{|}}\mu} =
0$}&$J^{\mu}\!\,_{{o\atop|}\mu} = 0$ &$F_{\mu\nu} = \delta_{\nu}C_{\mu} - \delta_{\mu}C_{\nu}$\\

\hline
{\bf I-Condition} & {$\,\overcirc{G}_{\mu\nu} = {\mathbb{E}}_{\mu\nu} + g_{\mu\nu}Z + A_{(\mu\nu)}$} &$F_{\mu\nu} = {\mathbb{F}}_{\mu\nu} +
N^{\beta}(\gamma_{\mu\nu\beta} +
\Lambda_{\beta\mu\nu}) + A_{[\mu\nu]}$ &
{${T}^{\mu}\,_{\sigma{o\atop{|}}\mu} = 0$}&
$J^{\mu}\!\,_{{o\atop|}\mu} = 0$&$F_{\mu\nu} = \delta_{\nu}C_{\mu} - \delta_{\mu}C_{\nu}$\\

\hline \multirow{2}{*}{{\bf General case}} &
{$\,\overcirc{G}_{\mu\nu} = {\mathbb{E}}_{\mu\nu} + g_{\mu\nu}(Z -
\frac{1}{2}Q) + {\mathbb{V}}_{(\mu\nu)} $} &$F_{\mu\nu} = {\mathbb{F}}_{\mu\nu} + N^{\beta}(\gamma_{\mu\nu\beta} +
\Lambda_{\beta\mu\nu}) + {\cal V}_{[\mu\nu]}$  &\hspace{0.6
cm}---&\hspace{0.6 cm}---
&$F_{\mu\nu} = \delta_{\nu}C_{\mu} - \delta_{\mu}C_{\nu}$\\
 & {$\hspace{1.1cm} + \,\, {\mathbb{U}}_{\mu\nu}$}&{$\hspace{1.1cm} + \, \,{\mathbb{U}}_{\mu\nu} +
{\mathbb{W}}_{\mu\nu}$} & & &{$ $}\\
 \hline
\end{tabular}

\vspace{2cm}
\begin{center}
Table 3\\
{\bf Vertical section}
\end{center}
\begin{center}

\begin{tabular}{|l|l|l|l|l|l|}
 \hline
\multirow{2}{*}{\bf Condition}&\multicolumn{2}{l|}{\hspace{0.2 cm}{\bf Gravitation and Electromagnetic tensors }}
&\multicolumn{2}{l|}{\bf Conservation laws}&\multirow{2}{*}{\bf Maxwell's Equations}\\
\cline{2-5} \cline{2-5}  & {\bf Gravitation  tensor} & {\bf Electromagnetic tensor} &\hspace{0.5 cm}{\bf MEC} &\hspace{0.2 cm}{\bf CDC}
&\\
\hline {\bf CB-Condition} &vanishes& vanishes
&{trivial}
&trivial
 &trivial\\
\hline
 {\bf B-Condition} &$\,\overcirc{G}_{ab} = {\mathbb{E}}_{ab}$  & $F_{ab} = {\mathbb{F}}_{ab}$ &{${T}^{a}\,_{b{o\atop{||}}a} = 0$} &
$J^{a}\!\,_{{o\atop||}a} = 0$&$F_{ab} = \dot{\partial_b}C_{a} - \dot{\partial_a}C_{b}$\\
\hline
{\bf C-Condition} & vanishes&vanishes&
{trivial}&trivial&trivial\\
  \hline
{\bf I-Condition} &{$\,\overcirc{G}_{ab} = {\mathbb{E}}_{ab}$} &$F_{ab} = {\mathbb{F}}_{ab}$ &
{${T}^{a}\,_{b{o\atop{||}}a} = 0$}&
$J^{a}\!\,_{{o\atop|}a} = 0$&$F_{ab} = \dot{\partial_b}C_{a} - \dot{\partial_a}C_{b}$\\

\hline {\bf General case} & {$\,\overcirc{G}_{ab} = {\mathbb{E}}_{ab}$}
&$F_{ab} = {\mathbb{F}}_{ab}$ &
{${T}^{a}\,_{b{o\atop{||}}a} = 0$}&
$J^{a}\!\,_{{o\atop|}a} = 0$&$F_{ab} = \dot{\partial_b}C_{a} - \dot{\partial_a}C_{b}$\\
\hline
\end{tabular}

\end{center}

\end{landscape}

\newpage

We end the paper with the following remarks and comments:
\begin{itemize}
\item In the mathematical framework of general relativity, there are no geometric objects that can be identified with
matter and other fields different from the gravitational field. General relativity provides a geometric description of spacetime and purely gravitational phenomena.
Other interactions are simply {\it not included}. Though the dynamics of macroscopic objects can be treated successfully in the
context of general relativity theory, the
impact of other interactions on gravity and the absence of dynamical objects describing these interactions may lead to
the emergence of serious problems. In other words, although other interactions do not play an explicit role on the dynamics
of large objects, we cannot neglect their impact on gravitational phenomena. Theories having a wider geometric structure
than Riemannian geometry are therefore needed to reformulate the structure of spacetime in order to include other fields. AP-geometry, with its sixteen degrees of freedom, is capable of treating electromagnetic phenomena on an
equal footing with gravitational phenomena. In fact, the GFT, constructed in the framework of AP-geometry, not only gives a single
geometric description of both fields, but also gives matter a geometric origin. Moreover, GFT has been able to solve some of the problems present in the context of general relativity theory such as 
the horizon and the flatness problems \cite{SS}. Since the GFT emerges as a special case of our constructed field theory,
it follows that these problems are actually solved in the framework of our new theory. Moreover, other interactions
may be potentially included in the context of our new theory due to its wealth and its wide geometric scope. We hope that using the present theory,
constructed in the context of EAP-geometry, one would solve more of the general relativity problems.

\item In this paper, we have constructed a unified field theory in the framework of EAP-geometry. The field equations are obtained by a
variational method.  The formulated theory is a generalization of the GFT, in which the chosen Lagrangians are the horizontal and vertical analogues
of the Lagrangian used in the construction of the GFT. Five different interesting cases for the horizontal field equations have been singled out. The most general is derived under the mere assumption that
the nonlinear connection is independent of the horizontal fundamental vector fields. From this, follows
both the Integrability case and the Cartan-type case. The Berwald-type case is deduced independently.
Finally, under the Cartan-Berwald condition, the constructed field equations are shown to coincide with the GFT. On the other hand, the vertical field equations
are derived under no additional assumptions.

\item The nonlinear connection enters explicitly in the horizontal field equations (and implicitly in the vertical field equations). This is because the chosen horizontal Lagrangian is expressed in terms
of the torsion tensor field and its contraction which is, in turn,  expressed in terms of the canonical $d$-connection.
The latter is defined in terms of the horizontal fundamental vector fields and the
nonlinear connection, through the operator $\delta_{\mu}$. Accordingly, the mathematical structure of the constructed field theory relies heavily on the notion of the nonlinear connection. In fact, the splitting
of the field equations into horizontal and vertical counterparts is made possible only due to the existence of the nonlinear connection. 
The physical
role of the nonlinear connection, however, needs further investigation. This could be partially achieved if an appropriate physical interpretation of the directional argument $y$ is given.
One possible interpretation is the following: The vector $y^{a}$ is attached as an internal variable to each point $x$. Consequently, the $y$-dependence may be combined with the notion of anisotropy or non-locality \cite{V}.

\item The horizontal field equations in the Integrability case and the Cartan-type case {coincide}. On the other hand, the vertical field equations in the
Integrability case coincide with those obtained in the general case, while they disappear in the Cartan-type case.  Moreover, in
the Cartan-type case, it is possible to take into consideration the \lq\lq mixed\rq\rq \, Lagrangian
${\cal M} := |\lambda|g^{cd}M_{cd}$, where $M_{cd} := C^{\alpha}_{\epsilon c}C^{\epsilon}_{\alpha d} - C_{c}C_{d}$;
$C_{d} := C^{\alpha}_{\alpha d}$. The corresponding
field equations are then given by
$$\frac{1}{|\lambda|}\bigg(\frac{\delta {\cal M}}{\delta \,\undersym{\lambda}{j}_{\beta}}\bigg)
{\,\undersym{\lambda}{j}_{\sigma}} = 0.$$
The idea is that $C^{\alpha}_{\mu c}$ is one of the non-vanishing torsion tensors of the canonical $d$-connection under
the Cartan-type case, in addition to
the purely horizontal counterpart $\Lambda^{\alpha}_{\mu\nu}$, which is used in the construction of the horizontal field equations.

\item Though the curvature of the nonlinear connection does not in general vanish under the Berwald-type condition, {\it it does not contribute, in any way, to
the horizontal field equations obtained}. This is one of the reasons (besides the fact that the geometric objects involved are
functions of the positional argument $x$ only) why the horizontal field equations coincide with those of the GFT.

\item In a forthcoming paper, a linearization of the obtained field equations is carried out, in which more physical interpretations of the coordinates
$(x^{\mu}, y^{a})$ are given. We interpret $x^{1}$, $x^{2}$ and
$x^{3}$ as space coordinates, whereas $x^{4}$ is taken as the time
coordinate. On the other hand, the vector $y^{a}$ is attached as an
{\bf internal variable} to each point $x^{\mu}$, as previously stated. In this sense,
$y^{a}$ may be regarded as the spacetime fluctuations (micro
internal freedom) associated to the point $x^{\mu}$ \cite{GLS},
\cite{V}. Accordingly, the vertical field equations, we conjecture,
may express some kind of micro (or {\bf quantum}) phenomena.

\item The Cartan-type case is, roughly speaking, the {closest} natural generalization of the GFT (at the expense of killing
the vertical counterpart)\footnote{We show in our next paper that
the field equations in the Cartan-type case actually {\bf coincide}
with the GFT in the {\bf first} order of approximation.}, whereas
the horizontal counterpart of the Berwald-type case is {
identical} to the GFT (leaving untouched the vertical counterpart).
Since both electromagnetism and gravity follow from the GFT in the
first order of approximation \cite{aax}, it follows that both
theories are an outcome of the horizontal field equations in the
Berwald-type case in the first order of approximation. This will be
discussed in detail in the above mentioned forthcoming paper.

\item The unified field theory constructed in the present work is a covariant one, i.e., it has an external symmetry. It is not, in general, gauge invariant. But its
skew-symmetric section has an internal symmetry, in addition to its external one. The special cases given by equations
(\ref{gage}), (\ref{gages}) and (\ref{gagex})
show that the skew-symmetric section of the horizontal field equations, in the Integrability, the Cartan and the Berwald cases,
are gauge invariant under the local gauge group $U(1)$ (abelian gauge group). Indeed, considering the Integrability and the Cartan cases, let
$C_{\mu}^{*}$ be a (generalized) gauge transformation of $C_{\mu}$, that is,
\begin{equation}C_{\mu}^{*}: = C_{\mu} - \delta_{\mu}\phi,\nonumber\end{equation} where $\phi$ is
an arbitrary smooth function of both $x$ and $y$. Then
\begin{eqnarray*}F^{*}_{\mu\nu} :&=& \delta_{\nu}C_{\mu}^{*} - \delta_{\mu}C_{\nu}^{*}\\
&=&\delta_{\nu}(C_{\mu} - \delta_{\mu}\phi) -  \delta_{\mu}(C_{\nu} - \delta_{\nu}\phi)\\
&=&(\delta_{\nu}C_{\mu} - \delta_{\mu}C_{\nu}) + [\delta_{\mu}, \delta_{\nu}]\phi\\
&=&\delta_{\nu}C_{\mu} - \delta_{\mu}C_{\nu} = F_{\mu\nu},\\
\end{eqnarray*}
\\[- 1.2 cm] where in the last step we have used the integrability condition, namely, $[\delta_{\mu}, \delta_{\nu}] = 0$. The gauge invariance of the Berwald case
may be proved in a similar manner, taking into account that the gauge transformation in this case has the form
\begin{equation}C_{\mu}^{*} := C_{\mu} - \partial_{\mu}\phi,\nonumber\end{equation} where $\phi$ is
an arbitrary smooth function of $x$ only, together with the fact that $[\partial_{\mu}, \partial_{\nu}] = 0$.
 This insures the capability
of the theory under the above mentioned conditions to describe electromagnetism. Moreover, equation (\ref{po}) implies that the skew-symmetric section of the vertical field equations is
also gauge invariant. To see this, note that the gauge transformation in this case takes the form
\begin{equation}C_{a}^{*} := C_{a} - \dot{\partial}_{a}\phi,\nonumber\end{equation} where $\phi$ is
an arbitrary smooth function of $x$ and $y$. Proceeding as we did in the evaluation of $F^{*}_{\mu\nu}$, with necessary changes,
the gauge invariance now follows by noting that $[\dot{\partial_a}, \dot{\partial_b}] = 0$.

\vspace{0.2 cm}
\hspace{0.5 cm} On the other hand, in the general case, (\ref{CUzx}) may imply that the skew-symmetric section of the horizontal field equations may have
a more general
internal local symmetry than $U(1)$, i.e. invariance under a non-abelian gauge group. More efforts are needed to probe other interactions,
besides gravity and electromagnetism, that the theory may be capable of representing. This, in turn, may
shed some light on our understanding of the role of the nonlinear connection in the present theory.

\item To sum up, our constructed field theory is a pure {\bf geometric} attempt to unify gravity and electromagnetism. The two fields are treated, in a comprehensive way, as
one entity. The theory is manifestly {\bf covariant.} Its underlying geometry is the EAP-geometry.
The symmetric part represents gravitation, while the
anti-symmetric part represents electromagnetism.\footnote{This does not exclude the possibility that the
theory could potentially describe other physical interactions.} Moreover, the anti-symmetric part of the
vertical field equations is gauge invariant. Finally, all {\bf physical objects} involved are expressed in
terms of the {fundamental} tensors of the EAP-space together with the {nonlinear
connection} $N$ defined on the space.

\end{itemize}

In conclusion, denoting our constructed field theory by UFT (Unified Field Theory), assuming $n = 4$, we have the following

\begin{center} {\bf Table 4: Comparison between GFT and UFT}\\[0.2cm]
\begin{tabular}
{|c|c|c|c|c|c|c|c|c|c|}\hline
&&&&\\
 {\bf Field}&{\bf Field}&{\bf No. of Field}
 &{\bf Field}&{\bf No. of Field}\\
 {\bf Theory}&{\bf Variables}&{\bf Variables}&{\bf Equations}&{\bf Equations}\\[0.2cm]\hline
 &&&&\\
{\bf GFT}&$\,\undersym{\lambda}{i}_{\mu}$&$16$&$E_{\mu\nu} = 0$&$16$\\[0.2 cm]\hline
&&&&\\
{\bf UFT}&$\,\undersym{\lambda}{i}_{\mu}; \,\,\,\undersym{\lambda}{i}_{a}$&$32$&$E_{\mu\nu} = E_{ab} = 0$&$32$\\[0.2cm]\hline
\end{tabular}
\end{center}

\end{mysect}


\bibliographystyle{plain}

\end{document}